\documentclass[lettersize,journal]{IEEEtran}
\usepackage{amsthm}
\usepackage{amsmath}
\usepackage{amssymb}
\usepackage{algpseudocode}
\usepackage{bm}
\usepackage{booktabs}
\usepackage{url}
\usepackage{cite}
\usepackage{color}
\usepackage{float}
\usepackage{graphicx}
\usepackage{mathrsfs}
\usepackage{multicol}
\usepackage{multirow}
\usepackage{stfloats}

\makeatletter

\floatstyle{ruled}
\newfloat{algorithm}{tbp}{loa}
\providecommand{\algorithmname}{Algorithm}
\floatname{algorithm}{\protect\algorithmname}
\usepackage{tikz}
\usetikzlibrary{calc}
\usetikzlibrary{patterns}

\theoremstyle{plain}

\theoremstyle{plain}

\theoremstyle{plain}

\theoremstyle{plain}
\newtheorem{thm}{\protect\theoremname}

\author{
Wenjie~Liu,~\IEEEmembership{Graduate~Student~Member,~IEEE,}
and~Panos~Papadimitratos,~\IEEEmembership{Fellow,~IEEE}\\
\thanks{W. Liu and P. Papadimitratos are with the Networked Systems Security Group, KTH Royal Institute of Technology, 114 28 Stockholm, Sweden.}
\thanks{Corresponding author: Wenjie Liu (e-mail: \textit{wenjieli@kth.se}).}
\thanks{This work was supported in part by the SSF SURPRISE cybersecurity project, the Security Link strategic research center, and the China Scholarship Council. The computations were enabled by resources provided by the National Academic Infrastructure for Supercomputing in Sweden (NAISS), partially funded by the Swedish Research Council through grant agreement no. 2022-06725. We would also like to acknowledge the work of the organizers of Jammertest 2024.}
}



\usepackage[acronym]{glossaries}
\newcommand{\newac}{\newacronym}
\newcommand{\ac}{\gls}
\newcommand{\Ac}{\Gls}
\newcommand{\acpl}{\glspl}
\newcommand{\Acpl}{\Glspl}

\newac{speb}{SPEB}{square position error bound}
\newac[plural=EFIMs,firstplural=Fisher information matrices (EFIMs)]{efim}{EFIM}{Fisher information matrix}
\newac{ne}{NE}{Nash equilibrium}
\newac{mse}{MSE}{mean squared error}
\newac{toa}{TOA}{time-of-arrival}
\newac{snr}{SNR}{signal-to-noise ratio}
\newac{lan}{LAN}{local area network}
\newac{psd}{PSD}{positive semidefinite}
\newac{pd}{PD}{positive definite}
\newac{wrt}{w.r.t.}{with respect to}
\newac{lhs}{L.H.S.}{left hand side}
\newac{wp1}{w.p.1}{with probability 1}
\newac{kkt}{KKT}{Karush-Kuhn-Tucker}
\newac{wlog}{w.l.o.g.}{without loss of generality}
\newac{mle}{MLE}{maximum likelihood estimation}
\newac{gps}{GPS}{Global Positioning System}
\newac{rssi}{RSSI}{received signal strength indicator}
\newac{mimo}{MIMO}{multiple-input multiple-output}
\newac{csi}{CSI}{channel state information}
\newac{fdd}{FDD}{frequency division duplexing}
\newac{ms}{MS}{mobile station}
\newac{bs}{BS}{base station}
\newac{d2d}{D2D}{device-to-device}
\newac{slnr}{SLNR}{signal-to-interference-leakage-and-noise-ratio}
\newac{ula}{ULA}{uniform linear antenna array}
\newac{pas}{PAS}{power angular spectrum}
\newac{mmse}{MMSE}{minimum mean square error}
\newac{zf}{ZF}{zero-forcing}
\newac{rzf}{RZF}{regularized zero-forcing}
\newac{as}{AS}{angular spread}
\newac{aod}{AOD}{angle of departure}
\newac{iid}{i.i.d.}{independent and identically distributed} 
\newac{sinr}{SINR}{signal-to-interference-and-noise ratio}
\newac{tdd}{TDD}{time-division duplex}
\newac{rvq}{RVQ}{random vector quantization}
\newac{rhs}{R.H.S.}{right hand side}
\newac{mrc}{MRC}{maximum ratio combining}
\newac{cdf}{CDF}{cumulative distribution function}
\newac{a.s.}{a.s.}{almost surely}
\newac{los}{LOS}{line-of-sight}
\newac{jsdm}{JSDM}{joint spatial division and multiplexing}
\newac{map}{MAP}{maximum a posteriori}
\newac{klt}{KLT}{Karhunen-Lo\`eve Transform}
\newac{lbe}{LBE}{link bargaining equilibrium}
\newac{se}{SE}{Stackelberg equilibrium}
\newac{uav}{UAV}{unmanned aerial vehicle}
\newac{nlos}{NLOS}{non-line-of-sight}
\newac{pdf}{PDF}{probability density function}
\newac{em}{EM}{expectation-maximization}
\newac{knn}{KNN}{$k$-nearest neighbor}
\newac{svd}{SVD}{singular value decomposition}
\newac{nmf}{NMF}{non-negative matrix factorization}
\newac{umf}{UMF}{unimodality-constrained matrix factorization}
\newac{rmse}{RMSE}{rooted mean squared error}
\newac{olos}{OLOS}{obstructed line-of-sight}
\newac{mmw}{mmW}{millimeter wave}
\newac{ber}{BER}{bit error rate}
\newac{rss}{RSS}{received signal strength}
\newac{lp}{LP}{linear program}
\newac{ufw}{U-FW}{unimodal Frank-Wolfe}
\newac{utf}{UTF}{unimodality-constrained tensor factorization}
\newac{fw}{FW}{Frank-Wolfe}
\newac{iot}{IoT}{Internet-of-Things}
\newac{mae}{MAE}{mean absolute error}
\newac{crb}{CRB}{Cram\'er-Rao bound}
\newac{aoa}{AoA}{angle of arrival}
\newac{wcl}{WCL}{weighted centroid localization}
\newac[plural=GNSS,firstplural=Global Navigation Satellite Systems (GNSS)]{gnss}{GNSS}{Global Navigation Satellite System}
\newac{gsm}{GSM}{global system for mobile communications}
\newac{imu}{IMU}{inertial measurement unit}
\newac{rbf}{RBF}{radial basis function}
\newac{msf}{MSF}{multi-sensor fusion}
\newac{lidar}{LiDAR}{light detection and ranging}
\newac{glrt}{GLRT}{generalized likelihood ratio test}
\newac{sdr}{SDR}{software-defined radio}
\newac{ap}{AP}{access point}
\newac{lte}{LTE}{Long-Term Evolution}
\newac{raim}{RAIM}{Receiver Autonomous Integrity Monitoring}
\newac{pvt}{PVT}{position, velocity and time}
\newac{npl}{NPL}{Neyman-Pearson lemma}
\newac{sop}{SOP}{signals of opportunity}
\newac{pca}{PCA}{Principal Component Analysis}
\newac{svm}{SVM}{Support Vector Machines}
\newac{ekf}{EKF}{extended Kalman filter}
\newac{roc}{ROC}{receiver-operating characteristic curve}

\setkeys{Gin}{width=1.0\columnwidth}

\makeatother

\providecommand{\corollaryname}{Corollary}
\providecommand{\lemmaname}{Lemma}
\providecommand{\propositionname}{Proposition}
\providecommand{\theoremname}{Theorem}

\newcommand{\revaddone}[1]{#1}
\definecolor{mycolor1}{rgb}{0.494117647058824,0.184313725490196,0.556862745098039}
\definecolor{mycolor2}{rgb}{0.466666666666667,0.674509803921569,0.188235294117647}
\definecolor{mycolor3}{rgb}{0.301960784313725,0.745098039215686,0.933333333333333}
\definecolor{mycolor4}{rgb}{0.929411764705882,0.694117647058824,0.125490196078431}
\definecolor{mycolor5}{rgb}{0.635294117647059,0.078431372549020,0.184313725490196}
\definecolor{mycolor6}{rgb}{0.8500,0.3250,0.0980}

\begin{document}
\title{GNSS Spoofing Detection Based on Opportunistic Position Information}

\maketitle


\setlength\parskip{0pt}

\begin{abstract}
The limited or no protection for civilian \ac{gnss} signals makes spoofing attacks relatively easy. With modern mobile devices often featuring network interfaces, state-of-the-art \ac{sop} schemes can provide accurate network positions in replacement of \ac{gnss}. The use of onboard inertial sensors can also assist in the absence of \ac{gnss}, possibly in the presence of jammers. The combination of \ac{sop} and inertial sensors has received limited attention, \revaddone{yet it shows} strong results \revaddone{on} fully custom-built platforms. We do not seek to improve such special-purpose schemes. Rather, we focus on countering \ac{gnss} attacks, notably detecting them, with emphasis on deployment with consumer-grade platforms, notably smartphones, that provide \emph{off-the-shelf opportunistic information} (i.e., network position and inertial sensor data). Our Position-based Attack Detection Scheme (PADS) is a probabilistic framework that uses regression and uncertainty analysis for positions. The regression optimization problem is a weighted mean square error of polynomial fitting, with constraints that the fitted positions satisfy the device velocity and acceleration. Then, uncertainty is modeled by a Gaussian process, which provides more flexibility to analyze how sure or unsure we are about position estimations. In the detection process, we combine all uncertainty information with the position estimations into a fused test statistic, which is the input utilized by an anomaly detector based on outlier ensembles. The evaluation shows that the PADS outperforms a set of baseline methods that rely on \ac{sop} or inertial sensor-based or statistical tests, achieving up to 3 times the true positive rate at a low false positive rate.
\end{abstract}

\begin{IEEEkeywords}
GNSS attack detection, secure localization, opportunistic information
\end{IEEEkeywords}

\glsresetall

\section{Introduction}
\Acpl{gnss} face a wide range of attack threats, with spoofing being particularly concerning, as it allows adversaries to manipulate the \ac{gnss} position and time. Real-world incidents, such as disrupting sensor fusion algorithms to cause crashes in autonomous vehicles \cite{SheWonCheChe:C20} and misnavigation of luxury yachts \cite{PsiHumSta:J16}, highlight the \revaddone{increasing threat} of \ac{gnss} attacks \cite{Goo:J22}. The \revaddone{rising} sophistication and accessibility of \ac{gnss} spoofing technology further intensifies these concerns \cite{PsiHum:B21}. In response, various standalone strategies have been proposed to counter and identify attacks, ranging from the implementation of cryptographic protocols \cite{WesRotHum:J12,FerRijSecSim:J16,AndCarDevGil:C17} to signal-level mechanisms \cite{PapJov:C08,BroJafDehNie:C12,ZhaTuhPap:C15,LiuCheYanShu:C21}. Cryptographic solutions require extensive updates to satellites and receivers, which are cost-prohibitive and hard to implement at scale. \revaddone{Signal-level solutions, relying on characteristics such as \ac{aoa}, may require specialized hardware (e.g., antenna arrays) and may be ineffective in challenging environments, such as urban canyons with significant multipath \cite{PapJov:C08,BroJafDehNie:C12}. Moreover, many consumer-grade devices do not integrate \ac{gnss} receiver with strong anti-jamming/spoofing features or do not expose the necessary low-level signal data (e.g., phase, correlations) through standard operating system APIs \cite{Gps:J24}.} 

Modern consumer-grade mobile platforms, such as smartphones and autonomous vehicles, offer a promising opportunity for \ac{gnss} spoofing detection. These devices are equipped with opportunistic information sources, including connectivity (i.e., Wi-Fi, cellular, Bluetooth signals) and onboard sensors such as \acpl{imu}. This combination of data sources enables robust attack detection by cross-checking \ac{gnss} \revaddone{positions} with alternative positioning methods. However, leveraging these data sources for spoofing detection is nontrivial due to their inherent limitations: network-based positions often have larger fluctuation errors than \ac{gnss} and \ac{imu}, while \ac{imu}-based systems suffer from cumulative error \cite{OliSciIbrDip:C19,SheWonCheChe:C20,GaoLi:J22}. Therefore, effectively integrating these information sources into a unified framework for spoofing detection remains an ongoing topic that requires further investigation. 

An increasingly explored approach involves leveraging \ac{sop} and \ac{imu} for \ac{gnss}-denied environments or spoofing detection \cite{KhaRosLanCha:C14,CecForLauTom:J21,MorKas:J21,MaaKas:J21,KasKhaAbdLee:J22,OliSciIbrDip:J22,GaoLi:J23b,KasKhaKhaLee:J24,BaiSunDemZha:J24}. For instance, \cite{OliSciIbrDip:J22,BaiSunDemZha:J24} utilize \ac{sop}, such as Wi-Fi or cellular signals, to localize devices or detect deviations from \ac{gnss}-provided positions, typically employing threshold-based binary decision models. \cite{KhaRosLanCha:C14,CecForLauTom:J21} incorporate \ac{imu} data for trajectory analysis to identify inconsistencies indicative of spoofing. \cite{MorKas:J21,MaaKas:J21,KasKhaAbdLee:J22,KasKhaKhaLee:J24} use both \ac{sop} and \ac{imu} with sophisticated custom hardware platforms to present an accurate alternative positioning result, although they do not look for the detection per se. While these methods have shown promise for localization in \ac{gnss}-denied environments, they are not readily applicable on consumer-grade platforms, e.g., smartphones. Hence, \revaddone{we direct} efforts towards the detection of \ac{gnss} attacks with the use of all available inputs (network positions and onboard sensor data) and a hardware-agnostic design. 

We propose a \emph{Position-based Attack Detection Scheme (PADS)} that integrates network-based positioning results and velocity, acceleration, and orientation from the onboard sensors, termed \emph{off-the-shelf opportunistic information}, within a probabilistic detection framework for \revaddone{widely-used} consumer-grade mobile devices. The key idea involves accounting for various types of noisy position sources with different update rates, the movement of the \ac{gnss} receiver, and designing a test statistic. The construction of the test statistic involves two steps, which we refer to as a combo of model-based and data-driven techniques. First, we establish a closed-form solution to describe the relationship between motion and position data, resulting in a motion-constrained regression problem. This connects short-term estimation via the receiver's motion with long-term estimation through network-based positions, effectively smoothing positions with motion model constraints. Second, we employ Gaussian process regression \cite{RasWil:B05} to model the uncertainty inherent in the smoothed positions. Then, we calculate a weighted sum of both positions and uncertainties into a unified Gaussian function for the ensemble-based anomaly detection \cite{Pev:J16,ZhaNasLi:J19}. \revaddone{Crucially, PADS is designed as a software-based detection layer that utilizes opportunistic information commonly accessible in modern platforms, making it deployable without requiring specialized hardware or low-level signal access.}

Building on our earlier work \cite{LiuPap:C23}, we incorporate a learning-based detector into the decision-making part. Furthermore, we apply our PADS to an existing dataset \cite{SheWonCheChe:C20} in enhanced network simulations of terrestrial infrastructures with real-world Wi-Fi and cellular layouts, and a newly collected dataset from consumer-grade Android phones under real \ac{gnss} attacks. Additionally, we present improved theoretical and experimental analyses to further assess the effectiveness of the scheme. 

The main contribution of this work is a \textbf{probabilistic framework for \ac{gnss} attack detection}: We combine network-based positions and onboard sensors within a probabilistic framework for \ac{gnss} position attack detection. Our \emph{Position-based Attack Detection Scheme (PADS)} can fuse position, velocity, acceleration, and orientation data from various sources with different accuracies and update rates. It provides a robust and interpretable detection outcome, as well as a recovered position using benign \ac{gnss}, networks, and onboard sensors.

In PADS, we contribute a formulation of \revaddone{a} \textbf{motion-constrained regression problem for position smoothing} that combines short-term trajectory smoothing via \ac{imu} with long-term stabilization from network-based positions. This fusion reduces position noise and prevents \ac{imu} drift. A Gaussian process regression is then employed to quantify positional uncertainty. They are computationally efficient, with polynomial-time complexity, and mathematical proofs are presented. In addition, we design \textbf{anomaly detection via a weighted test statistic} that incorporates position trajectories and their uncertainties, enabling unsupervised anomaly detection. This hyperparameter-free approach uses ensemble methods to detect spoofing attacks and is also extensible and compatible with signal-level detection methods by incorporating signal properties, such as Doppler shift. 

We also contribute a \textbf{comprehensive evaluation on consumer-grade platforms} for \ac{gnss} attack detection. First, we evaluate with the help of a simulated autonomous driving platform. Second, we experiment with various Android smartphones, including different brands, prices, and chips (from Exynos, MTK, Qualcomm, and Google Tensor). \ac{gnss} attack strategies include a variety of meaconing and spoofing, incurring gradual deviation or position jumping, notably in a real-world setting, with data collected in Jammertest 2024. 

The rest of the paper is organized as follows: Sec.~\ref{relwor} provides background and reviews related work on \ac{gnss} attacks, detection, and network-based positioning. Sec.~\ref{sysmod} presents our system model and adversary. Sec.~\ref{profor} and \ref{prosch} detail the problem formulation and the proposed PADS. Sec.~\ref{numres} discusses evaluation and comparison with baselines. \revaddone{Finally, Sec.~\ref{conclu} concludes the paper.}

\section{Related Work and Background}
\label{relwor} 
\subsection{GNSS Attack and Detection}
\Ac{gnss} spoofing attacks typically craft fraudulent signals with precise power and format as per \ac{gnss} protocols. Before spoofing, the attacker may first employ jamming to deliberately disrupt \ac{gnss} signals, causing the victim to lose \revaddone{the} \ac{gnss} signal lock \cite{KasKhaAbdLee:J22}. Alternatively, with more sophisticated strategies, the attacker may gradually amplify the spoofing signal, eventually tricking the victim to follow it \cite{SheWonCheChe:C20}. Meaconing, the easiest method of spoofing signal generation, involves retransmitting rightful satellite signals from a different area. When it comes to authenticated \ac{gnss} signals, relay or replay attacks \cite{LenSpaPap:C22} can transmit satellite signals using low-complexity setups to deceive the victim into trusting the information. Another more sophisticated modification, known as selective delay \cite{PapJov:C08}, rebroadcasts separate satellite signals, allowing for modification of the position solution according to the attack scenario. Distance-decreasing (DD) attacks \cite{ZhaPap:C19b} provide additional options for adversaries, employing Early Detection and Late Commit to relay the \ac{gnss} signal, thereby making the relayed one appear to arrive earlier than it would have. 

\revaddone{Standalone detections often focus on analyzing the physical characteristics of \ac{gnss} signals, e.g., Doppler effect, \ac{aoa}, \ac{snr}, and \ac{rss} \cite{PapJov:C08,BroJafDehNie:C12,LiuCheYanShu:C21,ZhoLvDenKe:J22,SpaGeiPanPap:J25}. Recent advances include leveraging signal quality monitoring \cite{ZhoLvLiJia:J24} or machine learning on signal features \cite{IqbAmaSik:J24}. These methods can be effective against attacks that cause signal distortions, but mostly can not provide an alternative positioning. Moreover, access to the necessary low-level measurements (e.g., precise phase, \ac{aoa}) is often limited by hardware and standard APIs \cite{Gps:J24}. For example, \ac{aoa} typically requires specialized multi-antenna hardware. Furthermore, multipath in urban areas can also cause signal anomalies. Sophisticated attackers can also employ strategies such as slowly varying spoofing \cite{SheWonCheChe:C20,GaoLi:J22} or imitating \ac{aoa} that minimizes abrupt changes, making detection based solely on signals difficult.} 

Increasingly, \revaddone{low-end} \ac{gnss} receivers are integrated into mobile platforms that also feature onboard sensors and network interfaces. This presents opportunities for leveraging the sensors and networks to enhance attack detection. For instance, \cite{MicForCenTom:J22} focuses on \acpl{uav} to fuse \ac{gnss} with \ac{imu} data. Then, they use relative distance information obtained from \ac{rss} data to detect spoofing, with alternative positions based on multi-agent systems enhancing navigation robustness under spoofing conditions. \ac{sop} from terrestrial network infrastructures can assess \ac{gnss}-provided positions \cite{OliSciIbrDip:J22,BaiSunDemFen:J24}. This involves assuming adequate network scanning and the availability of \ac{bs} or \ac{ap} positions, using \ac{rss} or \ac{toa} as the distance measure between the mobile platform and the station to estimate position for checking \ac{gnss}. Additionally, in \cite{KhaRosLanCha:C14}, an \ac{ekf} integrates \ac{gnss} and \ac{imu} data, with \ac{raim} to assist with spoofing detection. Combined metric-based approach \cite{RotCheLoWal:J21} integrates multiple detection features such as autocorrelation distortion, \ac{rss}, pseudorange, carrier phase difference, and \ac{aoa} to enhance detection. 

Without considering \revaddone{the} detection of \ac{gnss} attack, \cite{MorKas:J21,MaaKas:J21,KasKhaAbdLee:J22,KasKhaKhaLee:J24} fuse \ac{sop} with \ac{imu} to provide position and navigation with great accuracy in \ac{gnss}-denied environments. The work considers and experiments with different grades of \acpl{imu}, types of devices, cellular clock errors, pseudorange measurement models, unknown transmitter locations, etc. However, most of \revaddone{these APIs} or information are still unavailable on consumer-grade platforms, e.g., smartphones. 

\subsection{Network-based Positioning}
In addition to the conventional dependence on \ac{gnss} for positioning, the network infrastructures, such as Wi-Fi, cellular, Bluetooth, and eLoran \cite{Gow:J20}, can also provide alternatives or backups for accurate localization. They play an important role in scenarios where \ac{gnss} signals may be limited or unavailable \cite{KasKhaAbdLee:J22}. 
Our focus is not to incorporate these positioning techniques into our framework, but to use off-the-shelf network-based positions to enhance the robustness and reliability of the detection process. 

Fingerprinting methods \cite{VoDe:J15,XinChe:J23,HuaLiuJiaDai:J23} are commonly used where databases of pre-collected fingerprints (\ac{rss}, magnetic field values, channel state information, or even visual information) are compiled. After that, deterministic or probabilistic fingerprint-matching algorithms are used for localization. They provide supplementary information for positioning or offering validations of \ac{gnss} attacks. However, fingerprint database construction is time-consuming, especially for wide-open outdoor environments. Hence, fingerprint-based positioning is often limited to a small area. 

Range-based methods \cite{LaoMorKimLee:J18,Mozilla2023} make use of various inputs such as \ac{rss}, propagation time, or \ac{aoa} to derive pseudoranges, which are then utilized for multilateration. Recent advances in network-\ac{gnss} hybrid positioning \cite{BaiSunDemZha:J22,SonDinZhaChe:J25} rely on the ranging information (e.g., \ac{toa} of 5G millimeter-wave) or localization results to provide additional observations in \revaddone{the} \ac{ekf} framework. 


\section{System Model}
\label{sysmod}
We consider a mobile \ac{gnss}-enabled platform \revaddone{that provides} computational power, heterogeneous network infrastructures and diverse sensors. At time $t$, the actual platform position, denoted as $\mathbf{p}_{\text{c}}(t) \in \mathbb{R}^2$, needs to be estimated based on positioning. $\mathbf{p}_0(t)$ represents the \ac{gnss} position at time $t$. Wi-Fi and/or cellular networks can provide positions, $\mathbf{p}_m(t)$, based on network-based positioning algorithms (e.g., \cite{MagLunPat:J15,ShaKas:J21,Mozilla2023}), where $m=1,2,...,M$ and $M$ signifies the number of network interfaces; $t$ spans from $1$ to $N$ seconds, with $N$ as the total number of time indexes. \Acpl{imu} provide multi-axis acceleration measurements, while velocity is possibly obtained from wheel speed sensors (e.g., in autonomous vehicles), denoted as velocity, $\mathbf{v}(t)$, acceleration, $\mathbf{a}(t)$, and orientation, $\boldsymbol{\omega}(t)$. Continuous network connectivity is not required, as some infrastructure might be temporarily inaccessible \revaddone{for} reasons independent of the mobile device itself. We assume that positioning errors of benign \ac{gnss}, Wi-Fi, and cellular networks are zero-mean Gaussian random variables \cite{MiuHsuCheKam:J15,JanTipPop:C16,SchGriTreHof:C03}. As the mobile platform moves along a path in benign conditions (as illustrated in Fig.~\ref{fig:locations}), the \ac{gnss}-derived position aligns with the opportunistic position information.

\begin{figure}
\begin{centering}
\includegraphics[width=\columnwidth]{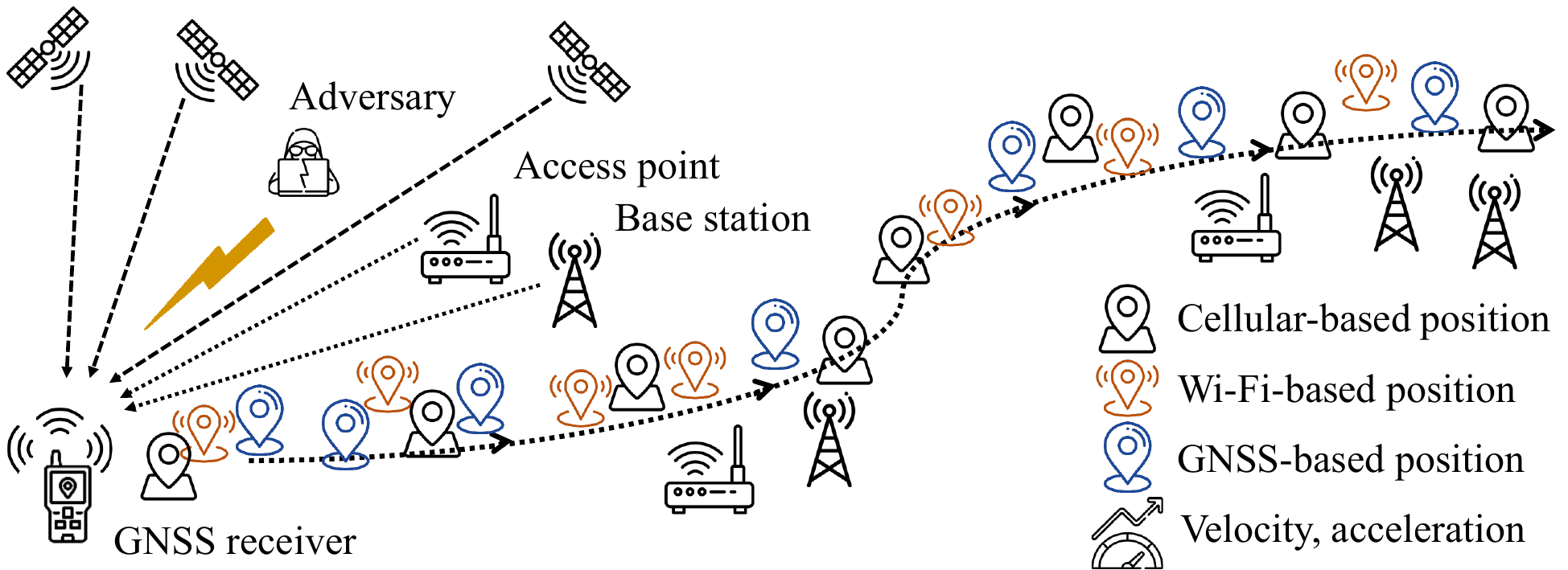}
\par\end{centering}
\caption{Illustration of position information from \ac{gnss} and network infrastructures, motion information from onboard sensors, and external adversary.}
\label{fig:locations}
\end{figure}

\textbf{Adversary:} Spoofed, relayed, or replayed \ac{gnss} satellite signals manipulate the mobile platform into falsely estimating the position. We assume that the attacker can observe the victim position and has access to \acpl{sdr} equipped with \ac{gnss} spoofing capabilities to falsify the signals. As \ac{gnss} signals are \revaddone{low power}, the attack can compel the victim to lose the lock on legitimate signals and then acquire the lock on adversarial signals. We remain neutral regarding the exact details of the attack, and do not restrict the type of attacker as long as it achieves the malicious alteration of position. In other words, as the \ac{gnss} solution includes position and time, detecting alterations solely in time, without corresponding changes in position, falls beyond the scope of our investigation. The attacker can skillfully design the victim's spoofed positions, considering its actual path. Certain trajectory designs, such as path drift \cite{NarRanNou:C19} and gradual deviation \cite{SheWonCheChe:C20}, remain almost undetectable for a period following the initiation of the attack (e.g., Kalman filter-based detection). 

The attacker considered here operates exclusively within the \ac{gnss} realm. \revaddone{We assume the opportunistic information sources (network-based position obtained via platform services, onboard sensor data) remain unaffected by the \ac{gnss} spoofing attack itself. This assumption is grounded in the practical separation of these systems (e.g., network signals are often authenticated or encrypted, and \ac{imu} sensors cannot be affected unless the device itself is compromised). Therefore, the scope of this work focuses on detecting \ac{gnss} attacks using trustworthy opportunistic data. Wi-Fi spoofing, cellular base station simulation, and physical sensor manipulation are considered beyond this work.} However, we assume that the adversary can extend periods of unavailability for network-based positioning by interfering with wireless networks. Furthermore, it is assumed that the attacker does not exert physical control over the victim, thereby preventing manipulation of the process for deriving position information from various network interfaces and onboard sensors. In short, the adversarial actions are confined to \ac{gnss} spoofing, while wireless networks may experience interference. As a result, during a spoofing attack, the \ac{gnss}-derived position should diverge from the actual position and not match the opportunistic position information. 

\section{Problem Formulation}
\label{profor}
Our objective is to assess the consistency between the \ac{gnss}-derived position and opportunistic position information from $\{\mathbf{p}_m(t),{\mathbf{v}}(t),{\mathbf{a}}(t),\boldsymbol{\omega}(t) \}$ to determine whether the current \ac{gnss} position is indicative of an attack. By evaluating the probability of a \ac{gnss} position attack, we seek to maximize the true positive rate of detection. Additionally, we aim for a detection scheme that remains reliable even if certain types of opportunistic information are unavailable. 

For the detection of \ac{gnss} position attacks at a given time $t$ and with $M$ network interfaces, we use data $\{\mathbf{p}_m(i),{\mathbf{v}}(i),{\mathbf{a}}(i),\boldsymbol{\omega}(i) \}$ for $0<i<t$ and $m \in \{0,1,...,M\}$ to determine whether $\mathbf{p}_0(t)$ is subject to an attack. Two corresponding hypotheses are presented as follows: 
\begin{itemize}
    \item $\mathcal{H}_0$: $\mathbf{p}_0(t)$ is not under attack;
    \item $\mathcal{H}_1$: $\mathbf{p}_0(t)$ is subject to attack.
\end{itemize}
Then, the decision made at the time $t$ is denoted as $\hat{\mathcal{H}}(t)\in \{\mathcal{H}_0,\mathcal{H}_1\}$. The true positive is expressed as $\hat{\mathcal{H}}(t)=\mathcal{H}_1$ under attack ($\mathcal{H}_1$), and the false positive is $\hat{\mathcal{H}}(t)=\mathcal{H}_1$ under $\mathcal{H}_0$. The true positive rate for $0<t \le N$ is 
\begin{equation}
    R_\text{TP} (\hat{\mathcal{H}}(t))=\mathbb{P}[\hat{\mathcal{H}}(t)=\mathcal{H}_{1}|\mathcal{H}_{1}].
\end{equation}
The false positive rate is 
\begin{equation}
    R_\text{FP} (\hat{\mathcal{H}}(t))=\mathbb{P}[\hat{\mathcal{H}}(t)=\mathcal{H}_{1}|\mathcal{H}_{0}].
\end{equation}

We define the \emph{detection time delay}, $\Delta T$, as the interval between the moment the alarm is raised and the start of the attack: 
\begin{multline}
    \Delta T=\min\left\{ t\Bigm|\mathbb{I}\{\hat{\mathcal{H}}(t)=\mathcal{H}_{1}|\mathcal{H}_{1}\}=1\right\} \\
    -\min\left\{ t\Bigm|\mathbb{I}\{\hat{\mathcal{H}}(t)=\mathcal{H}_{0}|\mathcal{H}_{1}\}=1\right\} 
\end{multline}
where indicator function $\mathbb{I}\{A|B\}$ equals to 1 when condition $A$ is satisfied given condition $B$.

The problem is to: (a) maximize $R_\text{TP}$ when fixing $R_\text{FP}$, (b) minimize $\Delta T$, and (c) provide a probability of \revaddone{being} under \ac{gnss} position attack, along with a recovered position that fuses opportunistic information to replace \ac{gnss} position when $\hat{\mathcal{H}}(t)=\mathcal{H}_1$.

\section{Proposed Scheme}
\label{prosch}
\begin{figure*}
\begin{centering}
\includegraphics[trim={0 0 0.6cm 0},clip,width=1.5\columnwidth]{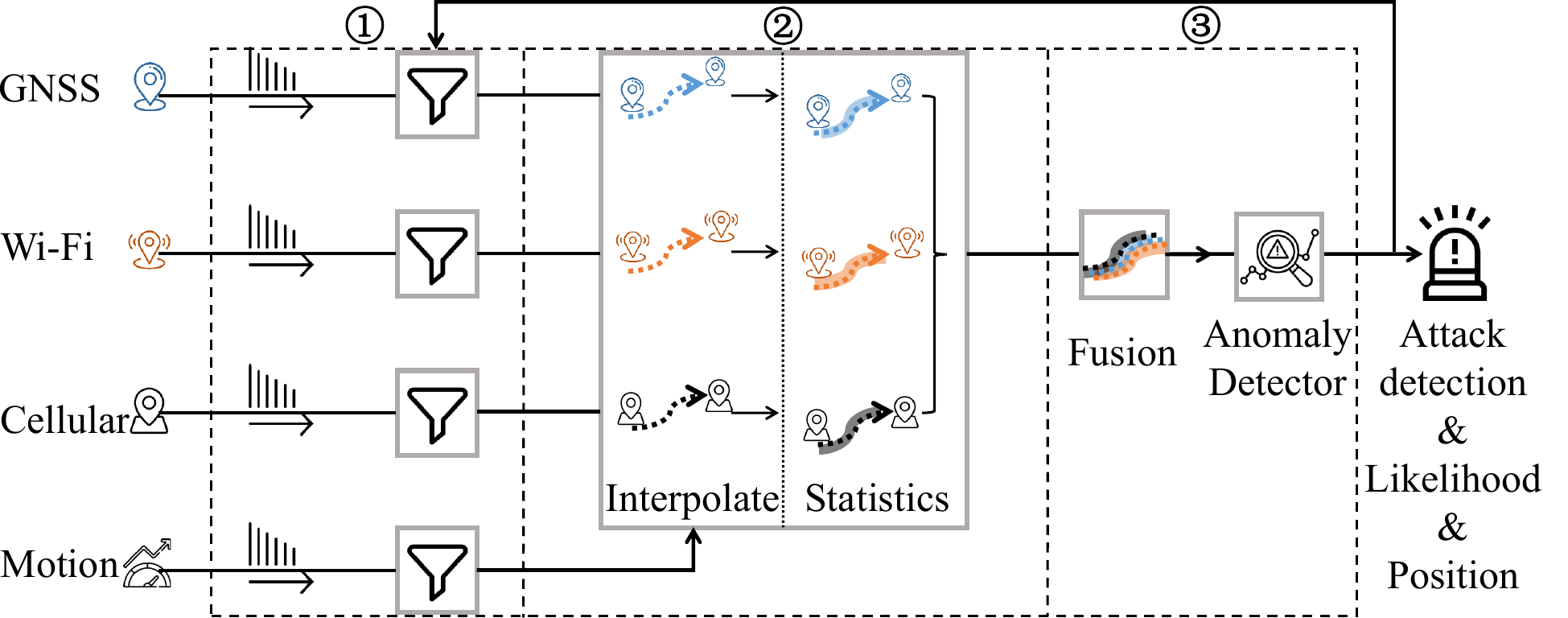}
\par\end{centering}
\caption{PADS overview. \revaddone{Inputs: \ac{gnss} position ($\mathbf{p}_0$), network positions ($\mathbf{p}_{1..M}$), and motion data ($\mathbf{v},\mathbf{a},\boldsymbol{\omega}$). Components: Rolling Window (\textcircled{\raisebox{-0.9pt}{1}})  selects recent data; Confidence Interval (\textcircled{\raisebox{-0.9pt}{2}}) estimates smoothed position $\hat{\mathbf{p}}_m$ and uncertainty $\hat{\boldsymbol{\Sigma}}_m$ for each source $m$ using motion-constrained regression and Gaussian processes; Decision-Making (\textcircled{\raisebox{-0.9pt}{3}}) fuses these intervals into a test statistic ($\varLambda_{1:M}$), computes an anomaly score ($f_{\mu,\sigma}$), makes a detection decision ($\hat{\mathcal{H}}$), and provides a recovered position ($\mu$). Outputs: $\hat{\mathcal{H}},f_{\mu,\sigma},\mu$ for each time $t$. It fuses information across multiple (available) sources and over the time domain.}}
\label{fig:scheme}
\end{figure*}
As Fig.~\ref{fig:scheme} shows, PADS detects attacks on \ac{gnss} positions by using information about network positions and movements, i.e., velocity, acceleration, and orientation. The input data is from \ac{gnss}, Wi-Fi, cellular, and onboard sensors. A rolling window takes a fixed-length series of positions, which will be interpolated to generate a smoothed trajectory. Subsequently, confidence intervals for these $M+1$ position series are calculated. The confidence intervals construct a fused test statistic to determine whether the current \ac{gnss} position reflects an attack, and if so, \revaddone{trigger} an alarm. The overall process is shown as Algorithm \ref{alg:overal}. 

\textbf{Rolling Window} (\textcircled{\raisebox{-0.9pt}{1}}) collects real-time positions of the platform from \ac{gnss}, Wi-Fi, and cellular sources (with $M+1$ categories available, and we consider $M=2$ here), in addition to velocity, acceleration, and orientation data from onboard sensors. These positions are organized in order of their timestamps. As a result, the filters implement rolling window techniques, taking fixed-length data at each detection time $t$ rather than utilizing the entire series.

\textbf{Confidence Interval} (\textcircled{\raisebox{-0.9pt}{2}}) is constructed by combining both motion and statistical models. The motion part (Sec.~\ref{subsec:motmod}) uses a local polynomial regression algorithm with movement constraints to fit the position data. It exploits both the short- and long-term characteristics of data: while onboard sensors provide high short-term accuracy, they are unable to maintain stable positional accuracy over time; to overcome this limitation, we integrate their data with less frequent but periodic positioning updates from terrestrial networks. We use regression to fit the positions and minimize the fitting error while adhering to movement constraints, ensuring the resulting fit conforms to velocity and acceleration. The statistical part (Sec.~\ref{subsec:stamod}) is a Gaussian process, \revaddone{focusing} on confidence intervals represented as a probability distribution, indicating the uncertainty of positions. We assume that the benign position series of $\mathbf{p}_m(t)$ follows a Gaussian process, with a mean already determined through the motion part. To compute the variance, Gaussian process regression uses a predefined covariance function and differences from fitted results. This allows for a better idea of the variability in position data. 

\textbf{Decision-Making} (\textcircled{\raisebox{-0.9pt}{3}}) builds a test statistic by using the mean and variance of confidence intervals. Then, anomaly detection relies on this test statistic across \revaddone{multiple sources and over the time domain}. To integrate the confidence intervals in the time domain, we calculate a weighted sum over time in Sec.~\ref{subsec:fusint}, with the weights normalized to ensure their summation to 1. The weighted sum is still Gaussian, with its mean and variance determined as a linear combination of the means and variances of the individual confidence intervals. To fuse the \ac{gnss}-derived position with $M$ opportunistic position sources, we multiply $M+1$ distributions for time $t$. Then, we create a fused test statistic based on it and apply an anomaly detector (Loda \cite{Pev:J16}) in Sec.~\ref{subsec:decmak}. 

\begin{algorithm}
\hspace*{\algorithmicindent} \textbf{Input} $\mathbf{p}_0(t),\mathbf{p}_1(t),\mathbf{p}_2(t),{\mathbf{v}}(t),{\mathbf{a}}(t),\boldsymbol{\omega}(t)$\\
\hspace*{\algorithmicindent} \textbf{Parameter} $w$\\
\hspace*{\algorithmicindent} \textbf{Output} $\hat{\mathcal{H}}(t),f_{\mu,\sigma}(t),\mu(t)$
\begin{algorithmic}[1]
\State \textbf{initialize} $S$ \Comment{Sequence of positions and motion data}
\State $t \gets 0$
\While{$t < N$}
\State $t \gets t+1$
\State \textbf{ensure} $length(S) = w$ \Comment{Rolling window}
\revaddone{
\State \textit{CI} $\triangleq \{\mathcal{N}(\hat{\mathbf{p}}_m(t), \hat{\boldsymbol{\Sigma}}_m(t))\}_{m=0..M} \gets$ Algorithm \ref{alg:confiden}
\Statex \hfill \Comment{Construct confidence intervals}
\State $\hat{\mathcal{H}}(t), f_{\mu,\sigma}(t), \mu(t) \gets$ Algorithm \ref{alg:decision}
\Statex \hfill \Comment{Detect attack using confidence intervals}
}
\If{$\hat{\mathcal{H}}(t)=\mathcal{H}_1$} 
    \State $S \gets \{S;\; \mathbf{p}_1(t),\mathbf{p}_2(t),{\mathbf{v}}(t),{\mathbf{a}}(t),\boldsymbol{\omega}(t) \}$
\Else
    \State $S \gets \{S;\; \mathbf{p}_0(t),\mathbf{p}_1(t),\mathbf{p}_2(t),{\mathbf{v}}(t),{\mathbf{a}}(t),\boldsymbol{\omega}(t) \}$
\EndIf \Comment{Append data at $t$ to the sequence}
\EndWhile
\end{algorithmic}
\caption{PADS overall operation with two alternative positioning sources and onboard sensors\label{alg:overal}}
\end{algorithm}

\subsection{Rolling Window for Detection}
\label{winrol}
We combine screening and detection: instead of analyzing the entire time series at each detection step for every time slot $t$, we select a fixed-size series by implementing a rolling window mechanism with a specific window size. This ensures that only recent network-based positions and motion data are used for evaluating potential \revaddone{attacks on the} current \ac{gnss}. 
\subsubsection{Coordinate Format}
Coordinates $\mathbf{p}_m(t) \in \mathbb{R}^2, m=0,1,...,M$ are formatted according to the World Geodetic System (WGS) Latitude, Longitude, Altitude standard. Data from onboard sensors $\mathbf{v}(t),\mathbf{a}(t) \in \mathbb{R}^3$ adheres to a local coordinate system in the same units as $\mathbf{p}_m(t)$, and $\boldsymbol{\omega}(t) \in \mathbb{R}^3$ comprises roll ($\phi$), pitch ($\theta$), and yaw ($\psi$) angles from gyroscope and magnetometer, which represent the orientations with respect to the local coordinates and WGS. 
\subsubsection{Window Size $w$}
The length of data series $S$ is a parameter: $S=\{ \mathbf{p}_m(i),{\mathbf{v}}(i),{\mathbf{a}}(i),\boldsymbol{\omega}(i) \}$ for $t-w<i<t$. Numerous approaches are available \revaddone{to determine} a good rolling window size. One example is to use a ``trial and error'' strategy that minimizes the \ac{mse} of positioning. It is a small-scale experiment with a range of window sizes and evaluating their performance on the validation set before detection, which is shown in our experiment results. Upon selecting an appropriate window size $w$, we can then move forward with processing $S$.
\subsubsection{Processing $S$}
At each $t$, the filtering process receives detection feedback regarding the current \ac{gnss} position, indicating whether it is potentially under attack. In the event of an alarm, the filtering process constructs $S$ for $t+1$ by incorporating data from sources excluding \ac{gnss}, denoted as ${ \mathbf{p}_m(t),{\mathbf{v}}(t),{\mathbf{a}}(t),\boldsymbol{\omega}(t) }, m \in {1,2,...,M}$. Conversely, if no attack is detected, the filtering updates $S$ using information from all available sources, denoted as ${ \mathbf{p}_m(t),{\mathbf{v}}(t),{\mathbf{a}}(t),\boldsymbol{\omega}(t) }, m \in {0,1,...,M}$.

\subsection{Constructing Confidence Intervals}
\label{concon}
\revaddone{This process involves two main steps: first, a model-based approach using motion-constrained regression to obtain smoothed positions, depicted in the dotted lines of Fig.~\ref{fig:gp_part}, and second, a data-driven approach using Gaussian processes to model position uncertainty, as depicted in the shaded areas of Fig.~\ref{fig:gp_part}. The process is described} in Algorithm~\ref{alg:confiden}.
\begin{algorithm}
\hspace*{\algorithmicindent} \textbf{Input} $S$\\
\hspace*{\algorithmicindent} \textbf{Output} \textit{CI}
\begin{algorithmic}[1]
\State $i \gets 0$
\While{$i < w$}
    \State $i \gets i+1$
    \State $\mathbf{W} \gets$ \eqref{eq:proall} \Comment{Compute the curve-fitting parameter}
    \State $\hat{\mathbf{p}}_m(t-w+i) \gets$ \eqref{eq:estpmt} \Comment{Smoothen positions}
    \State $\mathbf{x}_m(t-w+i) \gets$ \eqref{eq:gpresraw} \Comment{Residuals after smoothing}
\EndWhile
\State $\hat{\mathbf{x}}_m(t) \gets$ \eqref{eq:gplinest} \Comment{Estimate uncertainty at $t$ using residuals}
\State \textit{CI} $\gets$ \eqref{eq:conint} \Comment{Combine smoothed position and uncertainty}
\end{algorithmic}
\caption{Construct confidence intervals of positions\label{alg:confiden}}
\end{algorithm}

\begin{figure}
\begin{centering}
\includegraphics[width=\columnwidth]{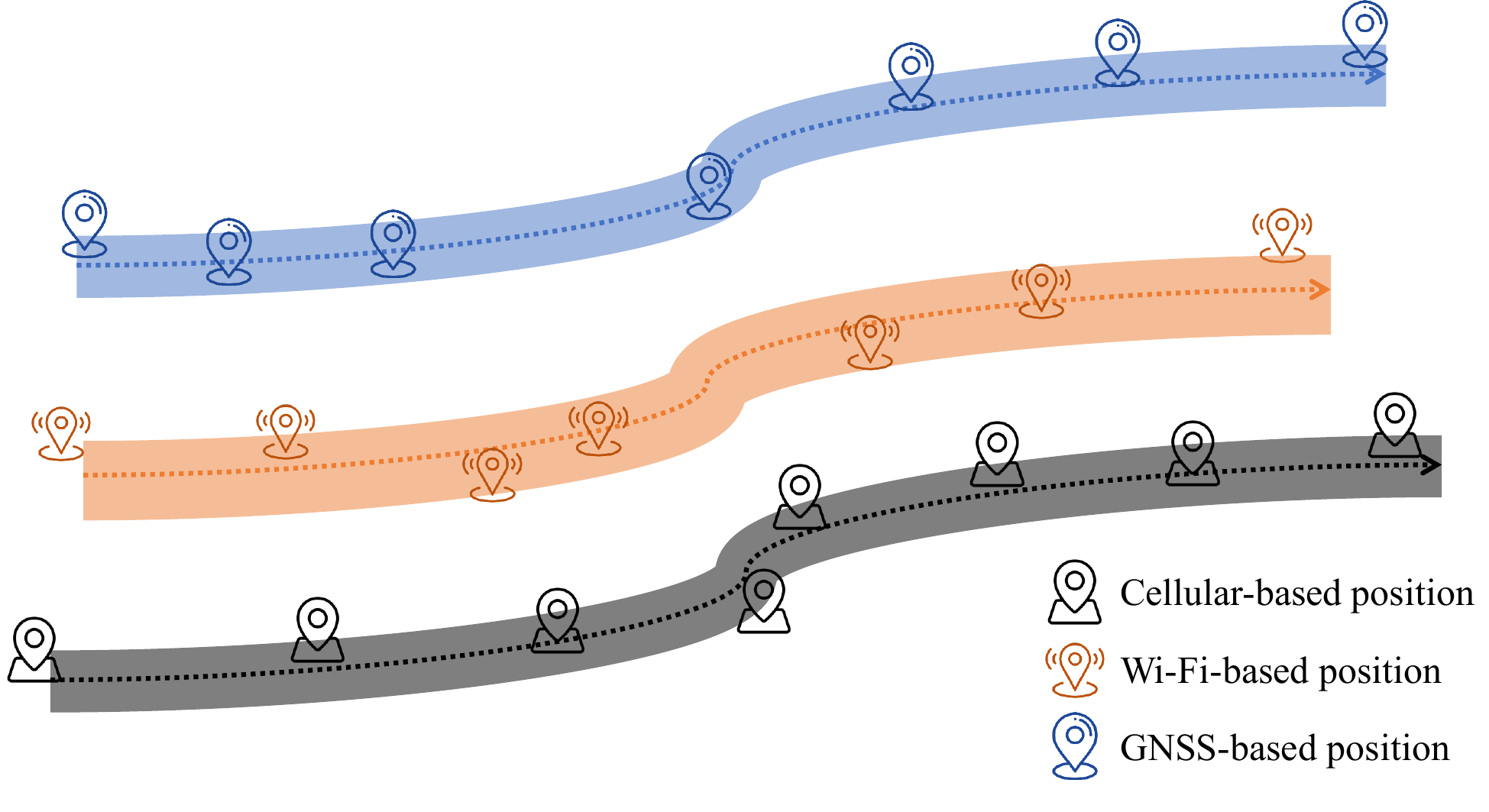}
\par\end{centering}
\caption{\revaddone{Local p}olynomial regression \revaddone{is used to estimate} traces (tiled view), and \revaddone{the} Gaussian process \revaddone{is used to model the} residual part of estimated positions (shaded areas).}
\label{fig:gp_part}
\end{figure}

\subsubsection{Motion-Assisted Fitting}
\label{subsec:motmod}
\revaddone{We use local polynomial regression for its flexibility in interpolating and predicting positions, based on discrete $\mathbf{p}_m(t)$ position points and motion data ${\mathbf{v}}(t),{\mathbf{a}}(t),\boldsymbol{\omega}(t)$. Crucially, we incorporate this with motion constraints derived from onboard sensors. It allows smoothing of noisy position measurements while ensuring the resulting position is physically plausible (adhering to velocity and acceleration limits), and the constrained optimization provides robustness against outliers compared to unconstrained fits or filters that might fuse spoofed data or \ac{imu} noise into updates. Local polynomial regression} involves fitting a Taylor expansion at a given point of a function through weighted least squares \cite{Fan:b96}. Thus, for a polynomial with degree $n$, at a given time $t$, the estimator $\hat{\mathbf{p}}_m(t)$ is represented as
\begin{equation}
    \hat{\mathbf{p}}_m(t) = \mathbf{W}\mathbf{t}
\label{eq:estpmt}
\end{equation}
where $\mathbf{W} \in \mathbb{R} ^{2 \times (n+1)}$ denotes the polynomial coefficients of Taylor expansion that need to be determined, and $\mathbf{t}$ represents a $(n+1)$ dimensional vector, $[\mathbf{t}]_i=t^{i-1}$. 

To determine $\mathbf{W}$, we introduce an optimization problem and present a theorem for the estimation process, ensuring both \revaddone{computational efficiency and reliability} of position predictions. $\mathbf{W}$ at the $m$th position and time $t$ is from 
\begin{equation}
    \begin{array}{*{20}{c}}
    {\mathop {\min }\limits_{\mathbf{W}} }&{f_\mathcal{P}(\mathbf{W})}\\ 
    \textrm{s.t.}&{|\hat{\mathbf{p}}_m(t) - \tilde{\mathbf{p}}_m(t)| \le \boldsymbol{\epsilon}_t}
    \end{array}
\label{eq:proall}
\end{equation}
where $\tilde{\mathbf{p}}_m(t)$ is the $m$th position based on motion data (to constrain $\hat{\mathbf{p}}_m(t)$), and $\boldsymbol{\epsilon}_t \in \mathbb{R}^2$ represents a small tolerance. $\tilde{\mathbf{p}}_m(t)$ ensures the physical feasibility of a short-term movement (onboard sensors), and $\mathbf{p}_m(t)$ \revaddone{ensures} long-term (network-based positions) anti-spoofing considerations. The objective function for regression in \eqref{eq:proall} \revaddone{should minimize the weighted squared error between the fitted polynomial $\hat{\mathbf{p}}_m(t') = \mathbf{W}\mathbf{t'}$ and the observed positions $\mathbf{p}_m(t')$, defined as}
\begin{equation}
    f_\mathcal{P}(\mathbf{W})=\sum\limits_{t'=t-w}^{t} [\mathbf{W} \mathbf{t'}-\mathbf{p}_m(t')]^\mathrm{T} K_\text{loc}(t'-t)[\mathbf{W} \mathbf{t'}-\mathbf{p}_m(t')]
\label{eq:proobj}
\end{equation}
where $K_\text{loc}(x)=\exp(-\kappa x^2)$ is a kernel function assigning weights that help emphasize the contribution of \revaddone{recent} data points while down-weighting the influence of more distant points, $\kappa$ is a kernel parameter, and $\mathbf{p}_m(t')$ are data points.

To provide $\tilde{\mathbf{p}}_m(t)$ for \eqref{eq:proall} using motion data, it is essential to standardize the coordinate systems of onboard sensors. $\mathbf{R}$ represents the rotation matrix responsible for converting the local coordinate system to WGS coordinates \cite{LiuPap:C23}:
\begin{align*}\mathbf{R}(t) & =\mathbf{R}_{\psi}(t)\mathbf{R}_{\theta}(t)\mathbf{R}_{\phi}(t).
\end{align*}
The state of the mobile platform is $\big(\mathbf{p}_m(t),\mathbf{v}(t) \big)$, so
\begin{align}
\tilde{\mathbf{p}}_m(t)&=\mathbf{p}_m(t-\Delta t)+\mathbf{R}(t-\Delta t)\mathbf{v}(t-\Delta t)\Delta t\notag\\
&\qquad+\frac{1}{2}\mathbf{R}(t-\Delta t)\mathbf{a}(t-\Delta t)(\Delta t)^2
\end{align}
and $\tilde{\mathbf{p}}_m(0)$ are initialized by the first \ac{gnss} and network positions. Moreover, if the onboard sensor does not furnish velocity information, $\mathbf{v}(t)$ can be substituted by $\mathbf{v}(t-\Delta t)+\intop_{t-\Delta t}^{t}\mathbf{a}(t)\textrm{d}t$, and $\mathbf{v}(t-\Delta t)$ is from the checked \ac{gnss}. Similarly, if acceleration information is unavailable from the onboard sensor, $\mathbf{a}(t)$ is assumed to be zero, indicating uniform motion over $\Delta t$. Then, $\tilde{\mathbf{p}}_m(t)$ in the constraint provides a rough movement range for the smoothed position $\hat{\mathbf{p}}_m(t)$. 

\begin{thm}
The estimator $\hat{\mathbf{p}}_m(t)$ in \eqref{eq:estpmt} can estimate $\mathbf{p}_m(t)$ within polynomial time.
\end{thm}
\begin{proof}
See Appendix \ref{app:prothe1}.
\end{proof}

\revaddone{This guarantees that} \eqref{eq:proall} \revaddone{is convex and solvable in polynomial time, making} the estimation \revaddone{suitable for real-time applications}. Updating the position estimations from $t=1$ to $N$, the $\hat{\mathbf{p}}_m(t)$ values should closely resemble the dotted lines depicted in Fig.~\ref{fig:gp_part}, illustrating the smoothed positions. These lines represent the estimated trajectory of the mobile platform from a tiled view to enhance visualization.

\subsubsection{Modeling Uncertainty}
\label{subsec:stamod}
\revaddone{After obtaining the smoothed positions $\hat{\mathbf{p}}_m(t)$ from regression, we model the remaining uncertainty. While \eqref{eq:proobj} provides a measure of fit, we employ Gaussian processes for a more principled and flexible approach to modeling uncertainty. Gaussian processes \cite{RasWil:B05} offer a non-parametric, data-driven method to estimate the distribution of the position residuals. They can capture temporal correlations in the uncertainty via kernel functions, providing a better uncertainty representation than just assuming independent noise. Denote the residuals of the estimated positions at $m,t$ as}
\begin{equation}
    \mathbf{x}_m(t)=\hat{\mathbf{p}}_m(t)-\mathbf{p}_m(t).
\label{eq:gpresraw}
\end{equation}
Then, in the absence of \ac{gnss} attack-induced deviations, $\left\{\mathbf{x}_m(i);i\in (0,t)\right\}$ are zero-mean Gaussian random variables with unknown standard deviations $\boldsymbol{\sigma}_m (i)$. A covariance function $K(\mathbf{x}_m(t),\mathbf{x}_m(t'))=\frac{1}{2}\mathbb{E}[(\mathbf{x}_m(t)-\mathbf{x}_m(t'))^2]$ is selected to characterize the interrelation of two residuals, $\mathbf{x}_m(t)$ and $\mathbf{x}_m(t')$, at time $t$ and $t'$. Commonly used kernels include linear, polynomial, and squared exponential covariance functions, and the best model and hyperparameters can generally be selected from cross-validation \cite{RasWil:B05}. Subsequently, a linear unbiased estimator can estimate the residual, $\mathbf{x}_m(t)$: 
\begin{equation}
    \hat{\mathbf{x}}_m(t)=\sum_{i=t-w}^{t-1} \lambda_i \mathbf{x}_m(i)
\label{eq:gplinest}
\end{equation}
where $\sum_{i=t-w}^{t-1} \lambda_i=1$. Gaussian process regression calculates $\lambda_i$, minimizing the variance of the estimation error: 
\begin{equation}
    \begin{array}{*{20}{c}}
  {\mathop {\min }\limits_{\mathbf{\boldsymbol{\lambda}}} }&{\mathbb{V}[\hat{\mathbf{x}}_m(t)-\mathbf{x}_m(t)]} \\ 
  {\text{s.t.}}&{\sum_{i=t-w}^{t-1} \lambda_i=1} 
\end{array}
\label{eq:gpropt}
\end{equation}
which can be solved using the Lagrangian method. 
\begin{thm}
Given a covariance function, $\hat{\mathbf{x}}_m(t)$ in \eqref{eq:gplinest} can estimate $\mathbf{p}_m(t)$ \revaddone{uncertainty} in polynomial time.
\end{thm}
\begin{proof}
See Appendix \ref{app:prothe2}.
\end{proof}

As the prediction yields a distribution for each time $t$, the confidence intervals $\mathcal{I}_m(t)$ indicate the uncertainty of the estimation. Since \ac{gnss} and network-based positions are subject to observational noise, the confidence intervals conform to a Gaussian distribution at each time $t$ for each information source in a benign environment. Consequently, its mean $\hat{\mathbf{p}}_m(t)$ and standard deviation $\hat{\boldsymbol{\sigma}}_m (t)$ of $\hat{\mathbf{x}}_m(t)$ characterize the confidence interval:
\begin{equation}
    \mathcal{I}_m(t) \sim \mathcal{N}(\hat{\mathbf{p}}_m(t), \hat{\boldsymbol{\Sigma}}_m(t)),m=0,1,...,M
\label{eq:conint}
\end{equation}
where $\hat{\boldsymbol{\Sigma}}_m (t) = \operatorname {diag} ([\hat{\boldsymbol{\sigma}}_m (t)]^2) \in \mathbb{R} ^{2 \times 2}$ is a diagonal matrix and the square is Hadamard power. An illustration in Fig.~\ref{fig:gp_part} shows lines connecting the individual position pins, i.e., $\hat{\mathbf{p}}_m(t)$, while the shaded areas are the uncertainties, $\hat{\boldsymbol{\Sigma}}_m (t)$. 

\subsection{Decision-Making Using the Intervals}
\label{croseq}
 
\begin{algorithm}
\hspace*{\algorithmicindent} \textbf{Input} \textit{CI}\\
\hspace*{\algorithmicindent} \textbf{Parameter} $\gamma$\\
\hspace*{\algorithmicindent} \textbf{Output} \revaddone{$\hat{\mathcal{H}}(t),f_{\mu,\sigma}(t),\mu(t)$}
\begin{algorithmic}[1]
\State $\varLambda_{1:M} \left( \mathbf{p}_0 \left( t \right) \right) \gets$ \eqref{eq:fused_test} \Comment{Fuse confidence intervals}
\State \revaddone{$\mu(t) \gets$ \eqref{eq:poscom} \Comment{Fused alternative position}}
\State $f_{\mu,\sigma}(t) \gets$ \eqref{eq:anosco} \Comment{Compute anomaly score}
\If{$f_{\mu,\sigma}(t) \ge \gamma$} 
    \State $\hat{\mathcal{H}}(t) \gets \mathcal{H}_1$ \Comment{Positive as score exceeds a threshold}
\Else
    \State $\hat{\mathcal{H}}(t) \gets \mathcal{H}_0$ \Comment{Negative otherwise}
\EndIf 
\end{algorithmic}
\caption{Decision based on confidence intervals from opportunistic position information \label{alg:decision}}
\end{algorithm}

\revaddone{Having obtained Gaussian confidence intervals $\mathcal{I}_m(t) \sim \mathcal{N}(\hat{\mathbf{p}}_m(t), \hat{\boldsymbol{\Sigma}}_m(t))$, decision-making fuses this information} from all position sources (i.e., \ac{gnss}, Wi-Fi, and cellular-based positions) into a single test statistic. It then utilizes \revaddone{an} anomaly detector for \ac{gnss} position attacks. We have two perspectives in the context of test statistic construction. First, the temporal perspective assesses the historical behavior of positions over time to capture patterns and anomalies. Second, the categorical perspective groups different sources of positions. 
\subsubsection{Fusing Intervals}
\label{subsec:fusint}
To process the data, $S$, along with its associated confidence intervals, $\mathcal{I}_m(t)=\hat{\mathbf{p}}_m(t)+\hat{\mathbf{x}}_m(t)$, which are derived from Algorithm \ref{alg:confiden}, we fuse these confidence intervals. It involves aggregating the weighted confidence intervals across $t$ with weights denoted as $K(m,t)$, which is a kernel function to ensure that $K(m,t)$ from $t-w$ to $t$ sum to 1. Then, the temporal fusion is
\revaddone{
\begin{equation}
Z(m,t) \triangleq \sum_{t'=t-w}^{t}K(m,t')\mathcal{I}_m(t')
\end{equation}
and denote its \ac{pdf} as $f_{Z(m,t)}(\mathbf{p})$. Thus, the $m$th test statistic for $\mathcal{H}_0$ is 
\begin{equation}
\varLambda_m\left( \mathbf{p}_0 \left( t \right) \right) |\mathcal{H}_0
=f_{Z(m,t)}(\mathbf{p}_{0}(t)).
\end{equation}
where $\mathbf{p}_{0}(t)$ is \ac{gnss} position. }For $M$ sources of positions, the fused test statistic is 
\begin{equation}
    \varLambda_{1:M}\left( \mathbf{p}_0 \left( t \right) \right)=\prod_{m=0}^M \varLambda_m\left( \mathbf{p}_0 \left( t \right) \right) |\mathcal{H}_0 \;.
\label{eq:fused_test}
\end{equation}
To simplify the calculation, we observe that $\varLambda_{1:M}\left( \mathbf{p}_0 \left( t \right) \right)$ \revaddone{is proportional to a Gaussian \ac{pdf}}. 
\begin{thm}
$\varLambda_{1:M}\left(x \right) = \frac{S}{\sigma(t)}\varphi \left({\frac {x-\mu(t) }{\sigma(t) }}\right)$, where $S$ is a constant scaling factor,
\begin{equation}
    \sigma(t) = \left(\sum_{m=0}^M {\left(\sum_{t'=t-w}^{t}\left[K\left(m,t'\right)\right]^{2}\revaddone{\left[\hat{\boldsymbol{\sigma}}_m (t')\right]^2}\right)}^{-1} \right)^{-\frac{1}{2}}
\label{eq:varcom}
\end{equation}
\begin{equation}
    \mu(t) = \sigma^2(t) \sum_{m=0}^M \frac{\sum_{t'=t-w}^{t}K\left(m,t'\right)\hat{\mathbf{p}}_{m}(t')}{\sum_{t'=t-w}^{t}\left[K\left(m,t'\right)\right]^{2}\revaddone{\left[\hat{\boldsymbol{\sigma}}_m (t')\right]^2}} \;.
\label{eq:poscom}
\end{equation}
\end{thm}
\begin{proof}
See Appendix \ref{app:prothe3}.
\end{proof}
\revaddone{Note that in the fused mean $\mu(t)$, each source $m$'s contribution is inversely weighted by its estimated uncertainty $\hat{\boldsymbol{\Sigma}}_m(t')$ (derived from the Gaussian process). This means that sources estimated to be less certain (higher variance) are naturally down-weighted in the test statistic.}

\subsubsection{Decision-Making}
\label{subsec:decmak}
After constructing the function in \eqref{eq:fused_test}, we apply Loda \cite{Pev:J16} to generate an anomaly score $f_{\mu,\sigma}(t)$, i.e., the probability of \ac{gnss} being under attack, shown as Algorithm \ref{alg:decision}. \revaddone{It is unsupervised, requiring no labeled attack data for training; lightweight and efficient, based on an ensemble of simple histograms on projections, making it suitable for resource-constrained platforms; robust due to its ensemble nature; and hyperparameter-free, simplifying deployment.} The detector works by using random projections of the input and then comparing their histograms to find differences.

Input data consists of $\sigma(t)$ from \eqref{eq:varcom} and $\mu(t)-\mathbf{p}_0 \left( t \right)$ from \eqref{eq:poscom} (other possible information includes antenna gain, dilution of precision, etc.). For the training phase, the algorithm learns the benign behavior of the data by constructing a set of models, i.e., representations, using a subset of the available benign data. This can be summarized by the following steps: (i) projection: it projects benign data $(\sigma(t), \mu(t)-\mathbf{p}_0 \left( t \right))$ onto a lower-dimensional space using $k$ random projection vectors, $\{\mathbf{v}_i\}_{i=1}^k$, to get the projected ones, $\{\tilde{x}_i\}_{i=1}^k$, where $\tilde{x}_i=(\sigma(t), \mu(t)-\mathbf{p}_0 \left( t \right))\mathbf{v}_i$, (ii) histogram: it calculates histogram, $\mathbf{h}_i$, of each projected value, $\tilde{x}_i$, from $i$th projection vector, and (iii) representations: the projection vectors, $\{\mathbf{v}_i\}_{i=1}^k$, and histograms, $\{\mathbf{h}_i\}_{i=1}^k$, are the benign patterns. 

For the testing (detecting) phase, the algorithm uses the representations constructed during the training phase to detect anomalies: (i) projection: this step is the same as in the training phase to get the projected $\{\tilde{z}_i\}_{i=1}^k$, and (ii) comparison: the projected testing vectors are compared with the trained histograms $\{\mathbf{h}_i\}_{i=1}^k$ to compute the anomaly score, $f_{\mu,\sigma}(t)$, by using the frequency of $\{\tilde{z}_i\}_{i=1}^k$ in the distribution of $\{\mathbf{h}_i\}_{i=1}^k$:
\begin{equation}
    f_{\mu,\sigma}(t)=-\frac{1}{k}\sum_{i=1}^k\log\mathbb{P}[\tilde{z}_i].
\label{eq:anosco}
\end{equation}
Note that all hyperparameters are determined automatically, as presented in \cite{Pev:J16}; thus, the detector is termed hyperparameter-free. The anomaly score, $f_{\mu,\sigma}(t)$, reflects \revaddone{the degree of abnormality relative to the learned benign distribution}, which is empirically chosen based on \revaddone{the} benign training data and environment. When $f_{\mu,\sigma}(t) \ge \gamma$, the decision is to raise an alarm that \revaddone{the} \ac{gnss} position is attacked ($\mathcal{H}_1$). Even if it may be a false alarm, our provided recovered position from \eqref{eq:poscom} is close to \revaddone{the} actual position because \revaddone{it} is a secure fusion of \ac{gnss}, network positions, and onboard sensors that removes spoofed \ac{gnss} positions in $S$. Therefore, this will not endanger the operation of one system that relies on this \ac{gnss} position. 

\revaddone{
\subsection{Computational Complexity}
The complexity of the proposed PADS framework in Algorithm \ref{alg:overal} consists of solving the convex optimization problem \eqref{eq:proall}, Gaussian process regression \eqref{eq:gpropt}, and Loda \eqref{eq:anosco}. First, \eqref{eq:proall} implies computations for forming and solving a quadratic program. The cost to form the objective part is $\mathcal{O}(w(n+1)^2)$ and the constraint part is $\mathcal{O}(n)$, where $n$ is the polynomial degree, usually taking a value of $1-3$. The cost to solve it depends on the (analytic, numerical) method, and it is approximately $\mathcal{O}((n+1)^3)$. Considering $n$ is small, the total cost is $\mathcal{O}(w)$. }

\revaddone{Second, for the construction of uncertainty using a Gaussian process, the main cost lies in solving the linear system derived from the Lagrangian to find the weights $\lambda_i$. This typically involves inverting a $w \times w$ covariance matrix, leading to a complexity of $\mathcal{O}(w^3)$. Third, fusing intervals requires calculations with complexity linear in $M$ (number of sources) and $w$. Loda detector requires projecting the input data onto $k$ random vectors. The complexity per detection is $\mathcal{O}(M \times w + k)$. }

\revaddone{As a result, the total complexity of detecting $\mathbf{p}_0$ is $\mathcal{O}(w+w^3+M \times w + k)$. Hence, the complexity is dominated by the Gaussian process component ($\mathcal{O}(w^3)$), but remains polynomial. Given $w=15-30$, PADS is computationally feasible. Our following experiments on mobile platforms also confirmed efficient, real-time processing at typical \ac{gnss} rates.}

\section{Experiment Results}
\label{numres}

\begin{table}
\centering
\caption{Datasets for experiments.}
\begin{tabular}{lrrrr}
\toprule
 & Ground Truth & GNSS & Network & Onboard Sensor \\
\midrule
Dataset A & 1 Hz & 1 Hz & 1 Hz \& 1 Hz & 200 Hz \\
Dataset B & 1 Hz & 1 Hz & 0.1 -- 0.3 Hz & 100 Hz \\
\bottomrule
\end{tabular}
\label{tab:dataset}
\end{table}

\begin{figure}
\begin{centering}
\includegraphics[trim={0 0 0 0},clip,width=.5\columnwidth]{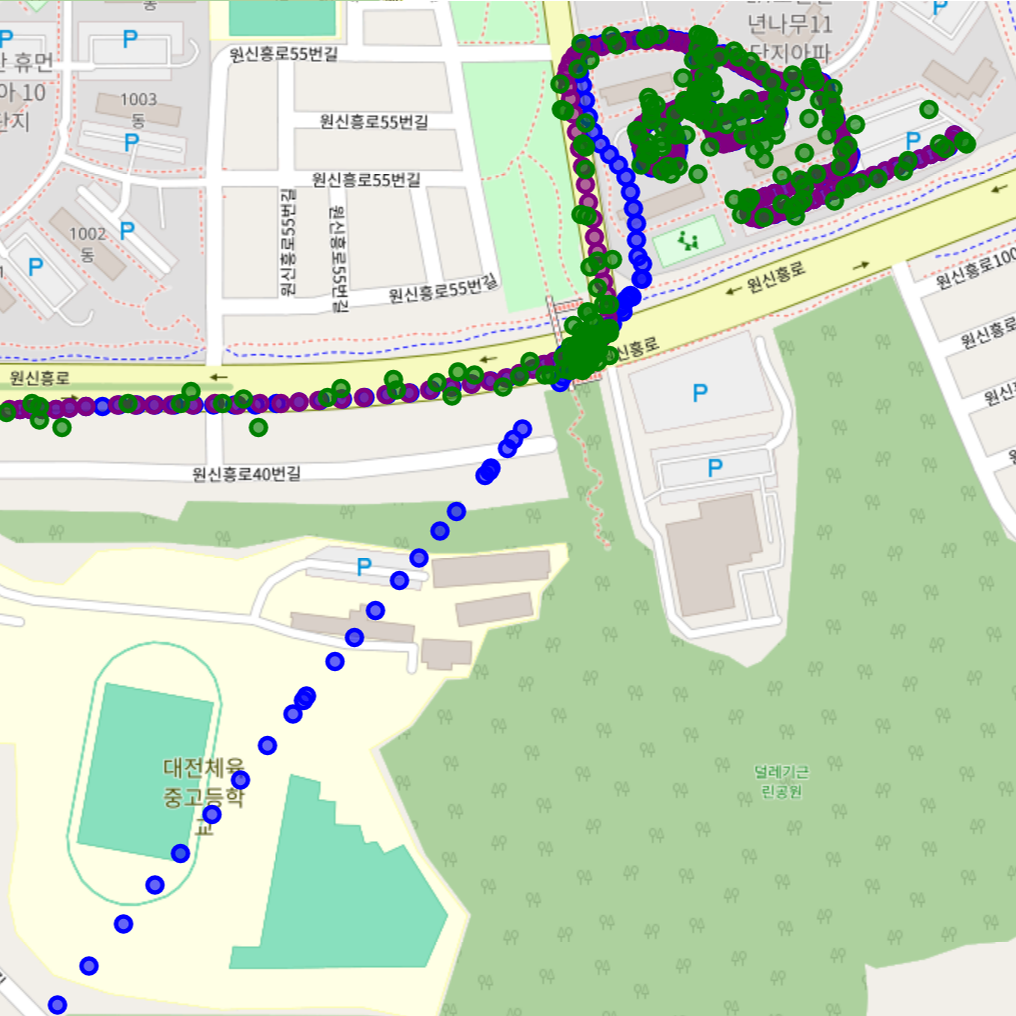}\includegraphics[trim={0 0 0 0},clip,width=.5\columnwidth]{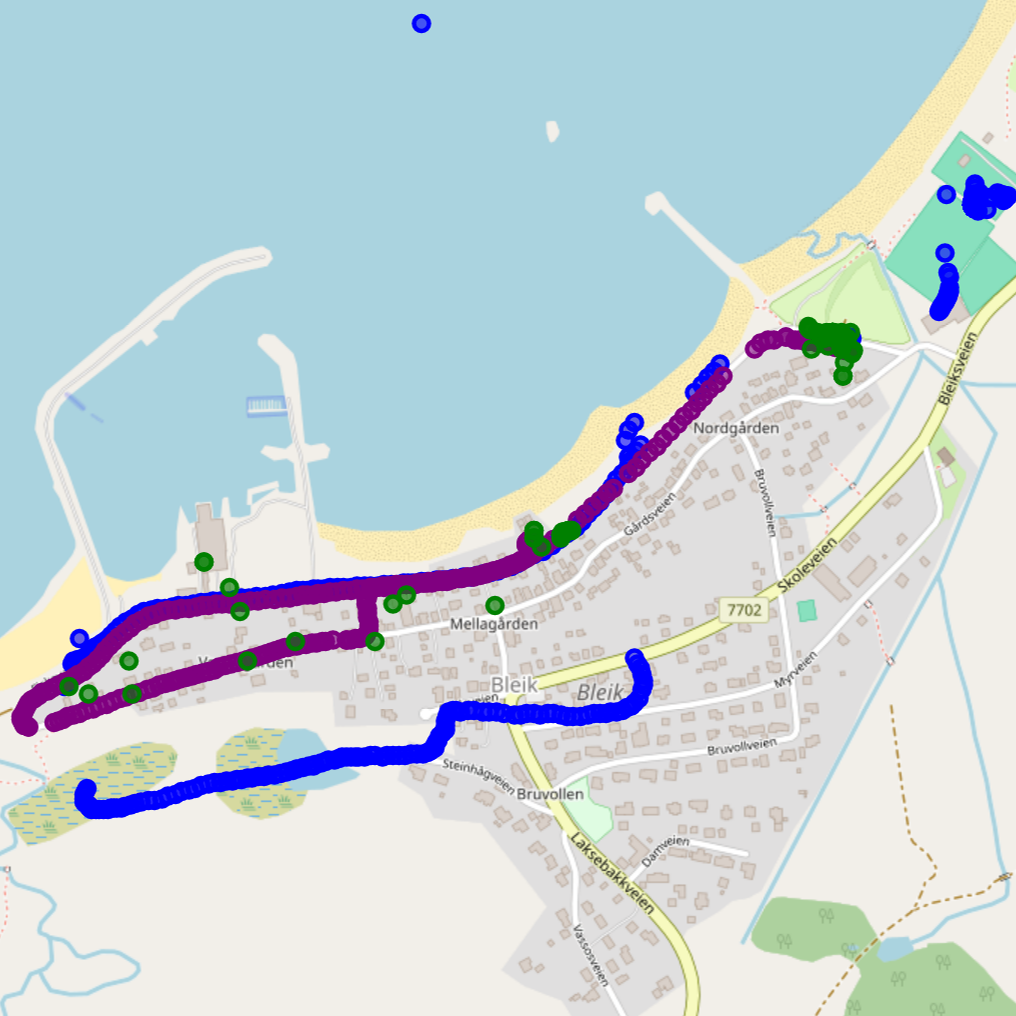}
\par\end{centering}
\caption{Ground truth (purple), network (green), and \ac{gnss} (blue) positions of two traces from Dataset A (left) and B (right). Note that ``network position'' refers to both Wi-Fi and cellular positions in Dataset A, and positions from Android Network Location Provider in Dataset B.}
\label{fig:dataset}
\end{figure}

\subsection{Experiment Setup}
We have two datasets (Table~\ref{tab:dataset} and illustrated in Fig.~\ref{fig:dataset}): (i) Dataset A \cite{SheWonCheChe:C20}, comprising 6 \ac{gnss} traces from outdoor urban environments and simulated attacks, and (ii) Dataset B is collected in a real \ac{gnss}-attack environment, Jammertest 2024, in Bleik village, Norway \cite{Jam:J24}. \Ac{gnss} receivers are mounted on vehicles with speeds ranging from 0 to 90 km/h. The error for benign \ac{gnss} positions is mostly within 3--10 meters for autonomous vehicles and smartphones in the test. 

\subsubsection{Dataset A}
It includes \ac{gnss} positions, \ac{imu} data, and ground truth positions from a simulated Apollo autonomous driving platform \cite{SheWonCheChe:C20}. In the same context, opportunistic position data is synthesized using a custom-made network simulator. The simulator generates the \ac{rss} from the seven nearest \acpl{bs} and \acpl{ap}, whose positions are from open databases \cite{Unw:J23} and \cite{BobArkUht:J23}. To accurately model the signal propagation, the simulation parameters for the free-space path loss model are derived from \ac{lte} TR36.814 \cite{TR36814} and 802.11n 2.4 GHz. The transmit power of a \ac{bs} is set at 20 dBm, while \revaddone{it is} 15 dBm for an \ac{ap}. Additionally, Gaussian noise with a variance level of 3 dB is added to the received power for network interfaces. The positioning algorithm based on \ac{ap} and \ac{bs} separately employs weighted nonlinear least squares \cite{Mozilla2023}, and the resulting positions are subject to estimation error with variances of 33 or 9 meters, as per \cite{KuuFalKatDia:J18,BaiSunDemZha:J22}. The unavailability probability, $U_m$, for $\mathbf{p}_m(t)$ is set at 0.2 per $t$ for all $m=1,2,...,M$, following a binomial distribution. The spoofing attacks are provided in \cite{SheWonCheChe:C20}. Its strategy consists of (i) vulnerability profiling, \revaddone{where} the attacker performs a constant spoofing to \ac{gnss} with a small deviation, and (ii) aggressive spoofing, \revaddone{where} the attacker makes the deviation grow exponentially after the receiver \revaddone{accepts} the spoofed position. 

\begin{figure}
\begin{centering}
\includegraphics[trim={1cm 1cm 0 1cm},clip,width=.53\columnwidth]{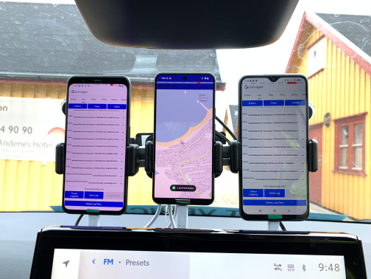}
\includegraphics[trim={5cm 3cm 0 0},clip,width=.439\columnwidth]{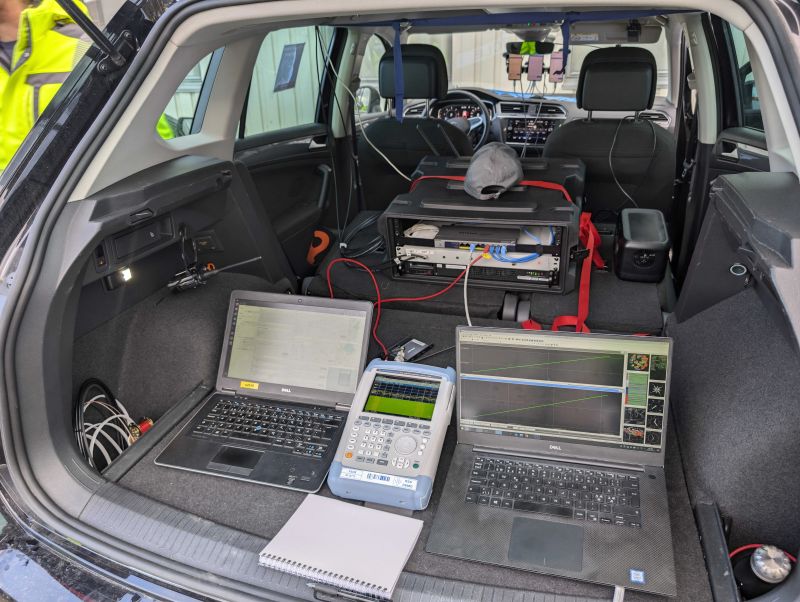}
\par\end{centering}
\caption{The placement of the mounted phones in a vehicle (car windshield).}
\label{fig:phones}
\end{figure}

\subsubsection{Dataset B}
\Ac{gnss} receivers are \revaddone{embedded} in Android smartphones, which are Samsung Galaxy S9, Redmi 9, Google Pixel 4 XL, and Google Pixel 8, shown in Fig.~\ref{fig:phones}, and two u-blox receivers as reference. The ground truth positions are provided by precise kinematic positioning results from u-blox ZED-F9P receivers using benign constellations and a nearby \ac{gnss} reference station\footnote{ZED-F9P logs UBX files, and then we use RTKLIB tool to convert \revaddone{them} into RINEX files. The reference station with \revaddone{the} code name ``ANDE00NOR'' is located at Andøya island and logs RINEX files. The RINEX files contain multiple constellations (\ac{gps}, GLONASS, Galileo, and BeiDou) plus SBAS and QZSS satellites, with two frequencies, so we use L1+L2, kinematic positioning mode, and benign constellations that exclude \ac{gps} and Galileo in the RTKLIB post-processing tool.}. The (opportunistic) network positions, \ac{gnss} positions, and \ac{imu} data are collected with GNSSLogger. The Android \texttt{LocationManager} provides access to multiple types of location services. It obtains 1 Hz updates of \ac{gnss} positions, while network positioning results are not at a fixed frequency and may update less frequently than \ac{gnss}, typically 3 to 10 seconds between updates, depending on the environment. The \texttt{SensorManager} stably delivers linear acceleration, gyroscope, magnetometer, and orientation at 100 Hz. \revaddone{Note that the u-blox receiver uses only benign satellite constellations and an RTK reference station as ground truth. The results presented later evaluate the detection capabilities of PADS using Android smartphone standard \ac{gnss}/Network/Sensors APIs while they were subjected to spoofing attacks. The u-blox by itself is not a baseline method. }The \ac{gnss} position attacks involved both synchronous and asynchronous spoofing \cite{Jam:J24}: stationary spoofing of small/large position jumps, SBAS spoofing, simulated driving, flying spoofing, as well as jamming. The attack equipment includes cigarette-type jammers, handheld jammers, permanently installed jammers, BladeRF x115 mobile \ac{sdr} spoofers, and USRP X300 \acpl{sdr}. Simulated \ac{gnss} signals corresponding to 5 predefined traces were transmitted using Skydel and USRP X300 with an amplifier, including paths with small deviations or jumping to another distant place\footnote{https://github.com/NPRA/jammertest-plan/blob/main/Testcatalog.pdf}. 

Since the attack gradually deviates \ac{gnss} positions from the ground truth positions, we define a lower bound of attack deviation, $\boldsymbol{\delta}_\text{d}$, for the ground truth detection results. Two types of positions are close when \revaddone{there is} no attack. If the distance between \ac{gnss} and ground truth is larger than $\boldsymbol{\delta}_\text{d}=10$ meters, we classify this \ac{gnss} position as the result of an attack, as our ground truth detection. 

\subsection{Baseline Methods}
We consider the following four methods from related work as baselines, including \ac{sop}, \ac{imu}-based, and both, which can be implemented with the collected opportunistic information. 
\subsubsection{Signals of Opportunity (SOP, Baseline 1)}
\cite{OliSciIbrDip:J22} uses the broadcast signals from the \acpl{bs} and \acpl{ap} to validate \ac{gnss}. It calculates weights $\mathbf{w}=[w_1,w_2,...,w_J]$ based on \ac{rss}, where $J$ is the number of \acpl{bs}/\acpl{ap}, and the estimated mobile platform position is the weighted centroid $\hat{\mathbf{p}}_\text{c}=\frac{\mathbf{w} \cdot \mathbf{p}_\text{bs}}{|\mathbf{w}|}$, where $\mathbf{p}_\text{bs} \in \mathbb{R}^{J \times 2}$ is the concatenated coordinate of all \acpl{bs} and \acpl{ap}. If the distance of $\hat{\mathbf{p}}_\text{c}$ and the \ac{gnss}-provided position is higher than a threshold, its outcome is ``under attack''. 
\subsubsection{Kalman Filter (KF, Baseline 2)}
\cite{KokHolSch:J18} fuses \ac{imu} and \ac{gnss} measurements, then we adapt it for \ac{imu}-based spoofing detection \cite{KhaRosLanCha:C14,CecForLauTom:J21}. The filter estimates the position of the mobile platform based on \ac{gnss} position and \ac{imu}. It minimizes the error of observation and motion to recursively get the mean and covariance matrix of the estimated position $\hat{\mathbf{p}}_\text{c}(t)$. If the residual between $\hat{\mathbf{p}}_\text{c}(t)$ and \ac{gnss} position is larger than a threshold, this method detects it as an attack. 
\subsubsection{Particle Filter (PF, Baseline 3)}
Similar to Kalman filter, we use a simple particle filter for \ac{imu}-based spoofing detection, which is based on the Markov Monte Carlo method \cite{Particle2019}. It generates particles uniformly around the initial position and then calculates the error between particles and position measurements. Then, the estimated position is a weighted sum of the particles based on errors. In the detection phase, a distance threshold-based detector is used for classification. 
\subsubsection{Combined Metrics (GLRT, Baseline 4)}
It follows the \ac{glrt} framework in \cite{RotCheLoWal:J21}. Step 1 involves calculating detection metric $\log \varLambda _m = -\frac{1}{2}||\mathbf{p}_0 \left( t \right)-\tilde{\mathbf{p}}_m(t)||^2_{\boldsymbol {\Sigma }_m (t)^{-1}}$ for $m$. Step 2 assumes these detection metrics \revaddone{are} statistically independent, and a likelihood ratio function is used to combine them: $\log \varLambda _{1:M} = \sum_{m=1}^M \log \varLambda _m$. Step 3 tests $\log \varLambda _{1:M}$ whether it is zero mean (i.e., under $\mathcal{H}_0$) or non-zero (i.e., an attack, under $\mathcal{H}_1$).

\textbf{Comparison} includes three metrics: true positive rate, $R_\text{TP}$, detection time delay, $\Delta T$, and the absolute error of recovered position, $\mu$. To perform a fair comparison between PADS and the baseline methods, we assess three cases and PADS variants: (i) exclusive utilization of network-based positioning results (removing the constraints of \eqref{eq:proall}, termed PADS-N), (ii) sole reliance on onboard sensors (removing network positions in \eqref{eq:proall}, termed PADS-O), and (iii) combined usage of network positions and onboard sensors (as \eqref{eq:proall}, termed PADS-A). 

\begin{figure}
    \centering
    \begin{minipage}{\columnwidth}
    \centering
    \includegraphics[trim={0 0 0 0},clip,width=\columnwidth]{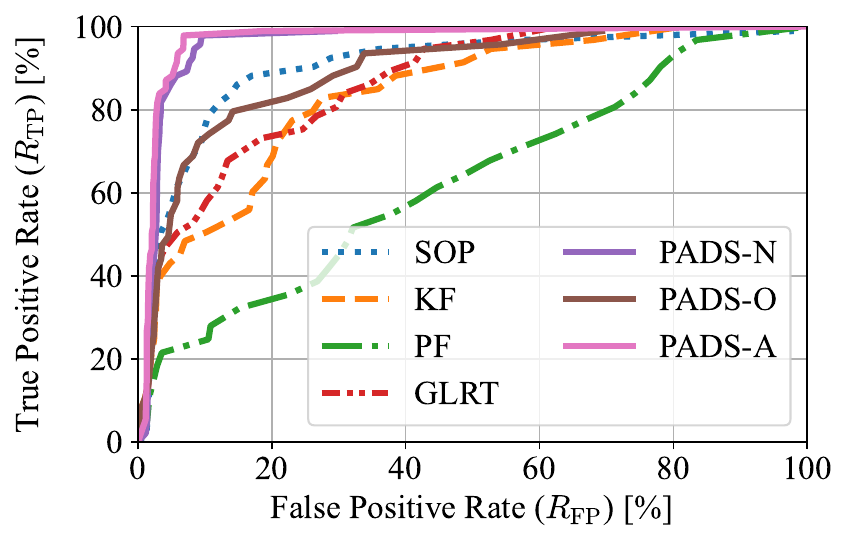}
    \end{minipage}
    \begin{minipage}{\columnwidth}
    \centering
    \includegraphics[trim={0 0 0 0},clip,width=\columnwidth]{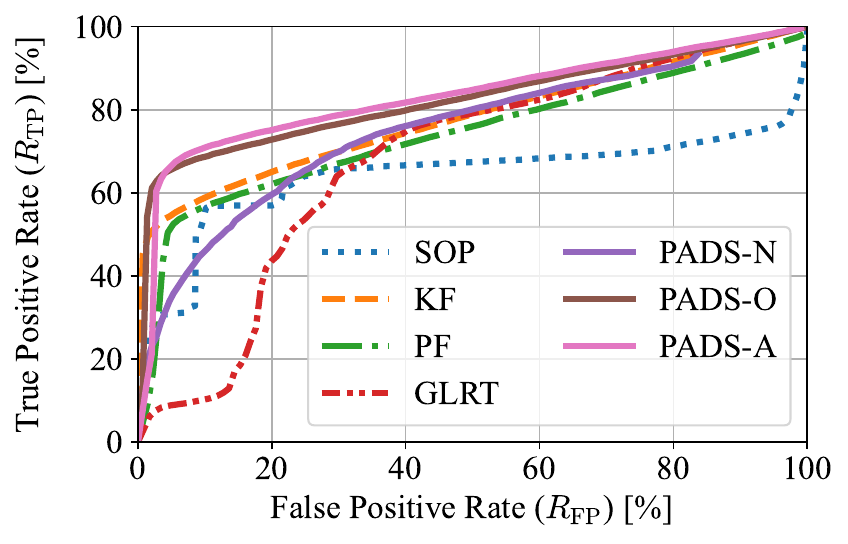}
    \end{minipage}
    \caption{\Ac{roc} curves for Dataset A (upper) and B (lower), i.e., $R_\text{TP}$ versus $R_\text{FP}$.}
    \label{fig:fprtpr}
\end{figure}

\subsection{Evaluation: True Positive Rate}
We investigate $R_\text{TP}$ at different $R_\text{FP}$ to plot \ac{roc} curves\footnote{All the work here is at the level of single-position detection, i.e., detecting each \ac{gnss} position based on historical measurements of the trace. We did not attempt to investigate whether a trace is under attack based on the spoofing detections of the positions of the entire trace.}. Our rolling window size is empirically set to $w=15$ for \ac{gnss}-provided and network-based positions, and we choose $\kappa=1$ in the kernel function $K_\text{loc}$.

PADS-N shares the same network conditions as Baseline 1. As shown in Fig.~\ref{fig:fprtpr}, PADS-N exhibits a modest improvement, at most 33\%, when $R_\text{FP}<15\%$. In Dataset A, it achieves $R_\text{TP}$ of 81--96\% when $R_\text{FP}$ is 5--15\%, compared to 63--87\% for Baseline 1. In Dataset B, it achieves $R_\text{TP}$ of 36--54\% when $R_\text{FP}$ is 5--15\%, similar to 33--57\% for Baseline 1. When $R_\text{FP}>10\%$, PADS-N has at most 20\% performance gain. Furthermore, the performance with Dataset A is better than \revaddone{that with} B because network positions are much sparser and \revaddone{noisier} in Dataset B. In general, as $R_\text{FP}$ increases, $R_\text{TP}$ tends to increase and converge for both PADS-N and \ac{sop}, both methods detecting the attacks well and thus \revaddone{resisting} gradual deviation or position jumping spoofing. 

PADS-O, KF (Baseline 2), and PF (Baseline 3) use \ac{gnss} position, incorporating motion data from onboard sensors. Fig.~\ref{fig:fprtpr} illustrates $R_\text{TP}$ as a function of $R_\text{FP}$. PADS-O and Baseline 2 maintain relatively consistent and similar trends of $R_\text{TP}$. However, in both Dataset A and B, Baseline 2 and Baseline 3 indicate relatively low $R_\text{TP}$, even at higher $R_\text{FP}$. PADS-O outperforms them with at most a 44\% $R_\text{TP}$ gain in Dataset A and 23\% in B, when $R_\text{FP}$ is 5--15\%. This is because the spoofed positions will influence the filters, while the proposed scheme detects and screens the spoofed position simultaneously. Furthermore, the regression in PADS-O can not only deal with Gaussian noise but also general zero mean noise, according to the least squares assumptions. Whenever the noise is not zero mean, PADS-O detects it as an anomaly caused by spoofing. 

Compared to PADS-A, \ac{glrt} (Baseline 4) struggles due to the absence of rolling window and motion-constrained regression for position data, hindering the effective fusion of heterogeneous data (i.e., positions, velocity, and acceleration). When utilizing all available opportunistic information sources, PADS-A consistently surpasses Baseline 4 across $R_\text{FP}$.

When comparing PADS-N and PADS-A, incorporating \ac{imu} results in a performance improvement of up to 7\% with Dataset A and 32\% with Dataset B. \ac{imu} data is particularly helpful when attack-induced position deviation \revaddone{grows} fast. In cases of subtle deviation changes, the effectiveness is relatively low. PADS-A achieves higher $R_\text{TP}$ compared to PADS-N and PADS-O, highlighting the performance gain resulting from the fusion of network-based positioning results and onboard sensor data. Considering PADS-O and other factors such as cost, sensor availability, and system complexity, the filters---Baseline 2 and Baseline 3---prove to be less effective compared to PADS-N but are relatively easy to acquire. We also observe that Baseline 2 and Baseline 3 perform much better in Dataset B than in A. This is because Dataset B is collected in a village without too many network infrastructures and the network positions are very sparse, so network positions can not produce high performance gain compared to onboard sensors. 

\begin{figure}
    \centering
    \begin{minipage}{\columnwidth}
    \centering
    \includegraphics[trim={0 0 0 0},clip,width=\columnwidth]{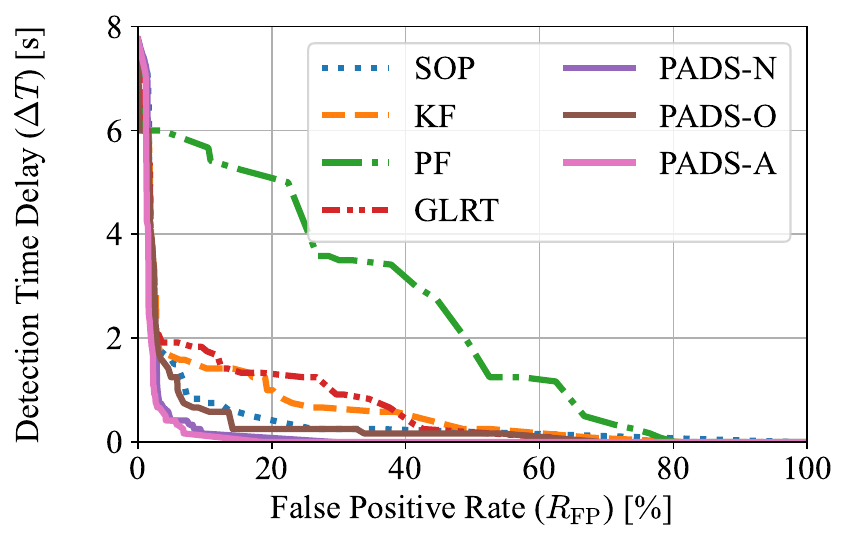}
    \end{minipage}
    \begin{minipage}{\columnwidth}
    \centering
    \includegraphics[trim={0 0 0 0},clip,width=\columnwidth]{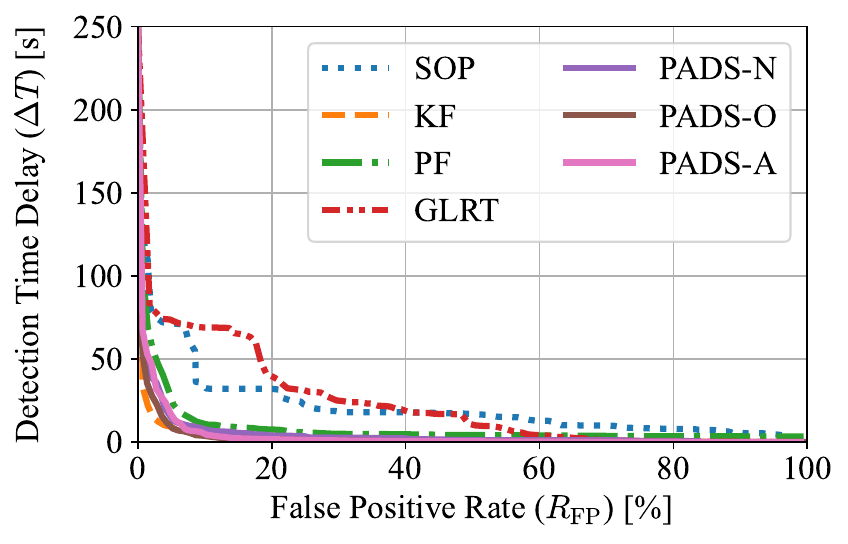}
    \end{minipage}
    \caption{$\Delta T$ versus $R_\text{FP}$ for Dataset A (upper) and B (lower).}
    \label{fig:fprdt}
\end{figure}

\subsection{Evaluation: Detection Time Delay}
The detection time delay, $\Delta T$, represents the duration between the initiation of an attack and its detection. Given the stealthy nature of spoofing attacks in \ac{gnss} traces, where deviations from the actual position evolve, analyzing the time delay in detecting such attacks is important. Our focus in this context is on measuring $\Delta T$, independent of computation delays. This metric reflects how fast the schemes can identify attacks, i.e., the sensitivity of the detection schemes. It is also influenced by factors such as the rolling window size, which will be investigated later.

Fig.~\ref{fig:fprdt} presents $\Delta T$ as a function of $R_\text{FP}$ for all three cases, respectively. PADS-N exhibits performance curves with similar shapes to Baseline 1 but mostly lower $\Delta T$. We have discovered that network-derived positions are inherently noisy, and PADS-N possesses the capability to effectively smoothen out this noise, resulting in a more accurate estimation of the actual position. The improved accuracy, in turn, allows it to detect spoofing and anomalies more quickly. Furthermore, PADS-N demonstrates a more significant reduction in $\Delta T$ as $R_\text{FP}$ increases to 15\%, indicating a better trade-off between false alarms and detection sensitivity.

PADS-O, KF (Baseline 2), and PF (Baseline 3) consistently exhibit $\Delta T$ exceeding at most 5--12 seconds when $R_\text{FP}$ is 5--15\%, due to the necessity of updating posterior distributions over sufficiently large time and accumulating deviation than the detection threshold. Moreover, the filtering process is influenced by the spoofed positions fed into the filter. In contrast, PADS-O outperforms both Baseline 2 and Baseline 3 by data fitting with motion information constraints and excluding spoofed positions before regression. 

PADS-A showcases lower $\Delta T$ compared to \ac{glrt} (Baseline 4) and the aforementioned filter-based schemes at equivalent $R_\text{FP}$. It consistently exhibits faster detection times across different datasets and $R_\text{FP}$. Additionally, PADS-A demonstrates a more significant reduction in $\Delta T$ as $R_\text{FP}$ increases, highlighting its success in keeping a balance between minimizing false alarms and reducing delay. Baseline 4 also performs well, delivering competitive $\Delta T$ compared to other schemes, especially at lower $R_\text{FP}<5\%$.

When considering three cases in conjunction with others, it is clear that better detection schemes and more precise network-based position information contribute to lower $\Delta T$. A high $R_\text{FP}$ also leads to a low $\Delta T$. Cases that give accurate opportunistic information for making decisions are more likely to catch spoofing attacks quickly. Conversely, methods that fail to integrate this information tend to miss detections, resulting in longer $\Delta T$. The use of the detection for high $R_\text{FP}$ is at the discretion of the method user, and we did not investigate how to use several successive alarms under some high $R_\text{FP}$. 

\subsection{Evaluation: Recovered Position Accuracy}
\begin{table*}
\caption{Absolute error evaluation of recovered position accuracy over different methods and datasets.}
\centering
\renewcommand{\arraystretch}{1.3}
\begin{tabular}{l*{8}{c}}
\toprule
\multirow{2}{*}{Methods} & \multicolumn{4}{c}{Dataset A} & \multicolumn{4}{c}{Dataset B} \\
\cmidrule(r){2-5} \cmidrule(l){6-9}
& Mean & Median & Best 20\% & Worst 20\% & Mean & Median & Best 20\% & Worst 20\% \\
\midrule
PADS-N & 2.45 & 1.57 & 0.48 & 3.77 & 243.82 & 30.90 & 8.48 & 260.36 \\
PADS-O & 11.43 & 2.02 & 0.59 & 5.00 & 276.61 & 35.87 & 9.32 & 350.17 \\ 
PADS-A & 2.92 & 1.87 & 0.83 & 4.22 & 260.61 & 35.62 & 9.31 & 345.05 \\ 
\ac{sop} & 5.87 & 4.80 & 2.42 & 8.85 & 544.91 & 298.14 & 283.60 & 655.61 \\ 
KF & 15.50 & 6.01 & 1.55 & 13.28 & 280.91 & 36.53 & 9.57 & 369.00 \\
PF & 64.11 & 24.11 & 6.92 & 81.96 & 266.64 & 48.52 & 9.00 & 275.06 \\
\ac{glrt} & 15.43 & 6.02 & 1.14 & 13.28 & 280.90 & 36.52 & 9.51 & 369.00 \\
\midrule
\revaddone{\ac{gnss}} & \revaddone{30.21} & \revaddone{22.33} & \revaddone{13.74} & \revaddone{41.41} & \revaddone{370.53} & \revaddone{169.93} & \revaddone{17.34} & \revaddone{479.69} \\
\bottomrule
\end{tabular}
\label{tab:recpos}
\end{table*}

The recovered position is defined as the mean of the confidence intervals in \eqref{eq:poscom}. We consider all spoofed positions and calculate the absolute error between the actual and the recovered position. Table~\ref{tab:recpos} showcases the worst 20\% (i.e., 20\% of the errors are higher than this), best 20\% (i.e., 20\% of the errors are lower than this particular), median, and mean error of the distribution. \revaddone{The last row, \ac{gnss}, is the result of positioning ($\mathbf{p}_0 \left( t \right)$) without any network signals and onboard sensors, representing the raw spoofing distances.} Baseline 1 has a relatively high worst 20\% error compared to its mean error, indicating \revaddone{a} big fluctuation in its performance. Notably, PADS-N demonstrates significantly lower error values than Baseline 1, particularly in terms of the worst 20\% error with both datasets. While PADS-O has a much lower error, Baseline 3 also has a relatively good performance for the best 20\% error. \revaddone{Regarding the raw \ac{gnss} error during attacks, the average distance is 30.21 meters in Dataset A and 370.53 meters in Dataset B.}

When considering the performance difference between Dataset A and B, all four types of errors are much higher with Dataset B than A. This is because the autonomous vehicle has much more accurate onboard sensors and \ac{gnss} positioning than the smartphones in the experiments. Also, autonomous vehicles \revaddone{are used} in the outdoor urban environment, whereas smartphones are \revaddone{used} in the village environment. As a result, the former has higher-quality network positions than the latter. Specifically, for the ratio of mean to median error, Dataset B is much larger than A, which means that many spoofed positions are not recovered, resulting in the mean being much larger than the median. This can also be confirmed in the detection accuracy of $R_\text{FP}$ and the worst 20\% error. This is due to some long-range absence of network positions in Dataset B, so these spoofed positions are mostly not detected and recovered. We also find the errors of filters and PADS-O are more similar in Dataset B than A, which means smartphone \ac{imu} contributes a lot in this experiment setting. 

PADS has the lowest mean error, ranging from 19\% to 98\% of the \ac{mae} observed in other methods, indicating better overall accuracy in estimation. This notable accuracy is due to the effective combination of various opportunistic information sources. While absolute error is a critical metric, it is also important to consider other related factors. They include consistency, reflecting whether the position result is stable as dataset size or categories of opportunistic information increase; variability, which denotes the uncertainty size associated with the position result; and specific application requirements, such as the types of sensors available on the mobile platform. 

\begin{table*}
\centering
\caption{$R_\text{TP}$ under different $w$ and $\kappa$ versus $R_\text{FP}$.}
\begin{tabular}{c*{7}{c}c*{7}{c}}
\toprule
\multirow{2}{*}{$R_\text{FP}$} & \multicolumn{7}{c}{$w$} & \multicolumn{7}{c}{$\kappa$} \\ 
\cmidrule(lr){2-8} \cmidrule(lr){9-15}
 & 5 & 10 & 15 & 20 & 25 & 30 & 35 & 0 & 0.5 & 1.0 & 1.5 & 2.0 & 2.5 & 3.0 \\ \midrule
5\% & 86.6\% & 90.3\% & 88.2\% & 89.0\% & 91.1\% & 91.3\% & 89.2\% & 55.4\% & 89.7\% & 88.2\% & 86.0\% & 85.5\% & 85.4\% & 86.0\% \\
10\% & 98.9\% & 98.3\% & 98.7\% & 97.8\% & 96.8\% & 98.7\% & 97.8\% & 73.7\% & 98.9\% & 98.9\% & 97.8\% & 97.8\% & 97.8\% & 97.8\% \\
15\% & 98.9\% & 98.9\% & 98.9\% & 98.9\% & 98.9\% & 99.8\% & 100\% & 84.9\% & 98.9\% & 98.9\% & 98.9\% & 98.9\% & 98.9\% & 98.9\% \\
20\% & 98.9\% & 100\% & 100\% & 100\% & 100\% & 100\% & 100\% & 91.4\% & 100\% & 100\% & 99.5\% & 100\% & 100\% & 99.5\% \\
25\% & 100\% & 100\% & 100\% & 100\% & 100\% & 100\% & 100\% & 91.4\% & 100\% & 100\% & 100\% & 100\% & 100\% & 100\% \\
\bottomrule
\end{tabular}
\label{tab:effwin}
\end{table*}

\subsection{Effect of Rolling Window}
We obtain some insights from the detection performance for different choices of window size $w$ and kernel parameter $\kappa$ in Sec.~\ref{winrol}. In this experiment, we use Dataset A, and \revaddone{the} settings are the same as the previous PADS-A, and $R_\text{FP}$ is 5--25\%, but $w$ ranges from 5 to 35 samples and $\kappa$ ranges from 0 to 3.0. The metrics for performance evaluation are the true positive rate, $R_\text{TP}$, shown in Table~\ref{tab:effwin}.

Determining the appropriate $w$ is a trade-off, as smaller windows offer computational efficiency but potentially sacrifice detection accuracy. Conversely, larger window sizes may lead to the processing of unnecessary historical data, resulting in slower and less accurate detection. Moreover, assigning weights to the data samples, which are controlled by $\kappa$ in the kernel function, is also important to make better use of the historical data. 

$R_\text{TP}$ in Table~\ref{tab:effwin} reveals some rough ``optimal'' values for $w$ and $\kappa$, indicating that further increasing $w$ or $\kappa$ does not significantly benefit detection performance. A small value of $w$ or $\kappa$ leads to lower accuracy, whereas increasing them will also increase $R_\text{TP}$ until they reach a sufficient size. This is because an oversized rolling window includes unnecessary data and gives too much weight, diminishing accuracy. $R_\text{TP}$ increases as $w$ increases, but the performance is stable after $w\ge10$. $R_\text{TP}$ increases as $\kappa$ increases, up until $\kappa=1.0$. Beyond this point, $R_\text{TP}$ starts to decrease again, as a value of $\kappa\ge1.5$ \revaddone{assigns} much higher weights to historical data, which can obscure the current information. In terms of computational complexity, the $w$ directly influences the complexity, i.e., the bigger $w$ comes with the higher complexity, as explained in Sec.~\ref{concon}. However, the change of $\kappa$ will not impact the complexity. 

\subsection{Discussion}
PADS demonstrates much better performance compared to baseline methods in Sec.~\ref{numres} due to its fusion of opportunistic information over time and consideration of position correlations. In most cases, by integrating data from all available sources, PADS-A naturally surpasses both PADS-N and PADS-O variants in $R_\text{TP}$ and $\Delta T$. Regarding the recovered position accuracy, PADS-A is in the middle of PADS-N and PADS-O. This highlights the difference between positioning and spoofing detection; even if the \ac{mae} of the positioning algorithm is excellent, it may not be good at detecting position spoofing. Hence, each has its own focus, and designing a detection-specific algorithm is important. Moreover, PADS-A mitigates the impact of accumulated errors from onboard sensors and one-time errors from network-based positioning, addressing both long-term and short-term inaccuracies. By incorporating uncertainty modeling (in Sec.~\ref{subsec:stamod}), we get an improved $R_\text{TP}$ compared to PADS without uncertainty. 

Given that the proposed PADS can operate at a software layer and does not rely on low-level hardware, its compatibility with other signal-level or cryptographic anti-attack methods is seamless. Incorporating this algorithm into existing consumer-grade devices is easy, making it convenient for deployment alongside other security measures \revaddone{to provide a multi-layer defense}. With positions and motion information as its only inputs, both commonly available in consumer-grade mobile platforms, the algorithm can opportunistically validate \ac{gnss} positions whenever network-based positions are accessible, providing a versatile and easily implementable solution for \ac{gnss} attack detection. Also, the computational complexity of PADS is adjustable, depending on the number of projection vectors in the lightweight anomaly detector \cite{Pev:J16} and the chosen window size in Algorithm \ref{alg:overal}. 

\section{Conclusion}
\label{conclu}
This paper presents an algorithmic framework for constructing confidence intervals for positions from opportunistic information and integrates them to estimate the probability of \ac{gnss} spoofing. It leverages both the motion dynamics of the mobile platform and the statistical characteristics of the positions. We employ a local polynomial regression technique with motion constraints, which is mathematically demonstrated to be convex for estimating and smoothing the position. Then, using Gaussian process regression, we capture the uncertainties inherent in position prediction and combine them into a fused test statistic for an unsupervised anomaly detector. \revaddone{The evaluation, based on both simulated and real-world attack datasets collected on common mobile platforms,} shows significant improvements, including up to a 54\% increase in the true positive rate. Specifically, when the false positive rate is between 5\% and 10\%, we have a 7--48\% gain in the true positive rate. 


\begin{appendices}
\section{Proof of Theorem 1}\label{app:prothe1}
\revaddone{
Recall the problem formulation \eqref{eq:proall}: 
\begin{equation*}
    \begin{array}{*{20}{c}} {\mathop {\min }\limits_{\mathbf{W} }}&{\sum\limits_{t'=t-w}^{t} [\mathbf{W} \mathbf{t'}-\mathbf{p}_m(t')]^\mathrm{T} K_\text{loc}(t'-t)[\mathbf{W} \mathbf{t'}-\mathbf{p}_m(t')]}\\ \textrm{s.t.}&{|\hat{\mathbf{p}}_m(t) - \tilde{\mathbf{p}}_m(t)| \le \boldsymbol{\epsilon}_t} \end{array}
\end{equation*}
where $\hat{\mathbf{p}}_m(t) = \mathbf{W}\mathbf{t}$ and $\mathbf{t}$ is the vector $[1,t,t^2,...,t^n]^\mathrm{T}$. $\mathbf{W}$ contains the coefficients we need to find. The kernel $K_\text{loc}(t'-t)$ assigns weights, and $\tilde{\mathbf{p}}_m(t)$ is the motion-derived position constraint.} 

\revaddone{To check for convexity, w}e take the second derivative of the objective function with respect to $\mathbf{W}$: 
\[
\revaddone{\nabla^2f_\mathcal{P}(\mathbf{W})= } 2 \sum\limits_{t=t'-w}^{t'}  K_\text{loc}(t-t')\cdot (\mathbf{t'}\cdot \mathbf{t'}^\mathrm{T} )^\mathrm{T} \otimes \mathbb{I}
\]
which is a positive definite matrix, as $K_\text{loc}(t-t')>0$ always holds\revaddone{, and the matrix $(\mathbf{t'}\cdot \mathbf{t'}^\mathrm{T} )^\mathrm{T} \otimes \mathbb{I}$ is positive semidefinite. Thus,} the objective function is convex. The constraints in \eqref{eq:proall} are equivalent to
$$
\left\{ \begin{array}{l}
	\mathbf{W} \mathbf{t} - \tilde{\mathbf{p}}_m(t) \le \revaddone{\boldsymbol{\epsilon}_t}\\
	\mathbf{W} \mathbf{t} - \tilde{\mathbf{p}}_m(t) \ge \revaddone{-\boldsymbol{\epsilon}_t}\\
\end{array} \right. ,\forall t
$$
which are \revaddone{the absolute difference between the fitted position $\hat{\mathbf{p}}_m(t) = \mathbf{W}\mathbf{t}$ and the motion-derived position $\tilde{\mathbf{p}}_m(t)$, so these are linear inequality constraints on $\mathbf{W}$. The set of points satisfying a system of linear inequalities forms a convex set.} Hence, \eqref{eq:proall} is a convex optimization problem. It is solvable using Lagrange multipliers \revaddone{or the practical algorithms like interior-point methods}; thus, the estimator $\hat{\mathbf{p}}_m(t)$ can estimate $\mathbf{p}_m(t)$ in polynomial time. 

\section{Proof of Theorem 2}\label{app:prothe2}
Ordinary Gaussian process regression uses a linear unbiased estimator for $\mathbf{x}_m(t)$. We can use Lagrange multipliers to extract the $\lambda_i$ parameters from the optimization problem. 
\begin{align*}
L(\boldsymbol{\lambda},\mu)&=\mathbb{V}[\hat{\mathbf{x}}_m(t)-\mathbf{x}_m(t)] +\mu(\sum_{i=t-w}^{t-1} \lambda_i-1)\\
&=\mathbb{E}[\sum_{i=t-w}^{t-1} \lambda_i \mathbf{x}_m(i)-\mathbf{x}_m(t)]^{2}+\mu(\sum_{i=t-w}^{t-1}\lambda_{i}-1)\\
&=\sum_{i=t-w}^{t-1}\lambda_{i}\mathbb{E}[\mathbf{x}_m(i)-\mathbf{x}_m(t)]^{2}\\
&\;-\frac{1}{2}\sum_{i,j}\lambda_{i}\lambda_{j}\mathbb{E}[\mathbf{x}_m(i)-\mathbf{x}_m(j)]^2+\mu(\sum_{i=t-w}^{t-1}\lambda_{i}-1)
\end{align*}
where $\mathbb{E}[\mathbf{x}_m(i)-\mathbf{x}_m(t)]^{2}$ and $\mathbb{E}[\mathbf{x}_m(i)-\mathbf{x}_m(j)]^2$ are calculated from a fixed covariance function $K(\mathbf{x}_m(t),\mathbf{x}_m(t'))$. Then, we take the partial derivatives of $L(\boldsymbol{\lambda},\mu)$ and set them to 0:
\begin{align}
\frac{\partial L(\boldsymbol{\lambda},\mu)}{\partial \boldsymbol{\lambda}}=0\\
\frac{\partial L(\boldsymbol{\lambda},\mu)}{\partial \mu}=0
\end{align}
obtaining a system of \revaddone{$w+1$} linear equations \revaddone{in the $w+1$ unknowns ($\boldsymbol{\lambda}$ and $\mu$)}. There exist several algorithms for solving it, such as Gaussian elimination\revaddone{. The computational complexity is dominated by the inversion or decomposition of the $(w+1)\times(w+1)$ matrix, which takes approximately $\mathcal{O}(w^3)$ arithmetic operations.}

\section{Proof of Theorem 3}\label{app:prothe3}
Assuming independence among the random variables $\mathcal{I}_m(t)$, we consider the time slots from $t-w$ to $t$, where $K(m,t)$ is determined by a kernel function to ensure that $K(m,t)$ from $t-w$ to $t$ sum to 1. Then, we use the moment-generating function: 
\begin{equation}
    M_{\mathcal{I}_m(t)}(s)=\mathbb{E}[e^{s\mathcal{I}_m(t)}].
\end{equation}
Recall that the weighted integral of $\mathcal{I}_m(t)$ from $t-w$ to $t$ is 
\begin{equation}
    Z(m,t)=\int_{t-w}^{t}K(m,t')\mathcal{I}_m(t')\textrm{d}t'
\end{equation}
practically utilized in discrete form: 
\begin{equation}
    Z(m,t)=\sum_{t'=t-w}^{t}K(m,t')\mathcal{I}_m(t').
\label{eq:z_mt}
\end{equation}
Its moment-generating function is 
\begin{align*}M_{Z\left(m,t\right)}\left(s\right) & =\mathbb{E}\left[e^{sZ\left(m,t\right)}\right]\\
 & =\mathbb{E}\left[e^{s\sum_{t'=t-w}^{t}K\left(m,t'\right)\mathcal{I}_{m}(t')}\right]\\
 & =\prod_{t'=t-w}^{t}\mathbb{E}\left[e^{sK\left(m,t'\right)\mathcal{I}_{m}(t')}\right]\\
 & =\prod_{t'=t-w}^{t}M_{\mathcal{I}_{m}(t')}\left(K\left(m,t'\right)s\right).
\end{align*}
The moment-generating function of a Normal distribution, $\mathcal{N}(\mu ,\sigma ^{2})$, is given by ${\displaystyle \exp({s\mu +{\frac {1}{2}}\sigma ^{2}s^{2}}})$. Thus, 
\begin{align*}
M_{Z\left(m,t\right)}\left(s\right)
=\prod_{t'=t-w}^{t}e^{K\left(m,t'\right)\hat{\mathbf{p}}_{m}(t')s+\frac{1}{2}\left[K\left(m,t'\right)\right]^{2}[\hat{\boldsymbol{\sigma}}_{m}(t')]^{2}s^2}\\
=e^{\sum_{t'=t-w}^{t}K\left(m,t'\right)\hat{\mathbf{p}}_{m}(t')s+\frac{1}{2}\sum_{t'=t-w}^{t}\left[K\left(m,t'\right)\right]^{2}[\hat{\boldsymbol{\sigma}}_{m}(t')]^{2}s^2}.
\end{align*}
$Z(m,t)$ follows a normal distribution, $\mathcal{N}(\sum_{t'=t-w}^{t}K\left(m,t'\right)\hat{\mathbf{p}}_{m}(t'),\sum_{t'=t-w}^{t}\left[K\left(m,t'\right)\right]^{2}\hat{\boldsymbol{\Sigma}}_{m}\left(t'\right))$, which means we compute the distribution $Z(m,t)$ by taking the weighted mean of distributions $\mathcal{I}_m(t)$.

Let $x\triangleq\mathbf{p}_{0}(t)$, $\mu_{(m)}\triangleq\mu\left(Z\left(m,t\right)\right)$, $\sigma_{(m)}\triangleq\sigma\left(Z\left(m,t\right)\right)$, and $\mu_{(0,1,...,M)},\sigma_{(0,1,...,M)}$ refer to the parameters of the multiplied Gaussian functions. The multiplication of the first two Gaussian functions, $\varLambda_0\left( \mathbf{p}_0 \left( t \right) \right) \times \varLambda_1\left( \mathbf{p}_0 \left( t \right) \right) |\mathcal{H}_0$, is 
\begin{align*}
    &\prod_{m=0}^1 \frac{1}{\sigma\left(Z\left(m,t\right)\right)} \varphi\left(\frac{x-\mu\left(Z\left(m,t\right)\right)}{\sigma\left(Z\left(m,t\right)\right)}\right)\\
    =&\prod_{m=0}^1 \frac{1}{\sigma_{(m)}} \varphi\left(\frac{x-\mu_{(m)}}{\sigma_{(m)}}\right)\\
    =&\frac{1}{2\pi\sigma_{(0)}^{2}\sigma_{(1)}^{2}}\exp\left[-\frac{(x-\mu_{(0)})^{2}}{2\sigma_{(0)}^{2}}-\frac{(x-\mu_{(1)})^{2}}{2\sigma_{(1)}^{2}}\right]\\
    =&\frac{S_{(0,1)}}{\sqrt{2\pi}\sigma_{(0,1)}}\exp\left[-\frac{(x-\mu_{(0,1)})^{2}}{2\sigma_{(0,1)}^{2}}\right]
\end{align*}
where \revaddone{$\varphi\left(\cdot\right)$ represents the standard normal \ac{pdf}, }$S_{(0,1)}$ is a constant scaling value, $\frac{1}{\sigma_{(0,1)}^{2}}=\frac{1}{\sigma_{(0)}^{2}}+\frac{1}{\sigma_{(1)}^{2}}$ and $\frac{\mu_{(0,1)}}{\sigma_{(0,1)}^{2}}=\frac{\mu_{(0)}}{\sigma_{(0)}^{2}}+\frac{\mu_{(1)}}{\sigma_{(1)}^{2}}$. Similarly, for $m$th Gaussian function, we have
\begin{align*}
    \frac{1}{\sigma_{(0,1,...,m)}^{2}}&=\frac{1}{\sigma_{(0,1,...,m-1)}^{2}}+\frac{1}{\sigma_{(m)}^{2}}\\
    \frac{\mu_{(0,1,...,m)}}{\sigma_{(0,1,...,m)}^{2}}&=\frac{\mu_{(0,1,...,m-1)}}{\sigma_{(0,1,...,m-1)}^{2}}+\frac{\mu_{(m)}}{\sigma_{(m)}^{2}}
\end{align*}
\revaddone{This process can be extended to multiply $M+1$ Gaussian likelihoods.} By mathematical induction, it follows that
\begin{align*}
    \frac{1}{\sigma^{2}}&\triangleq\frac{1}{\sigma_{(0,1,...,M)}^{2}}=\frac{1}{\sigma_{(0)}^{2}}+\frac{1}{\sigma_{(1)}^{2}}+...+\frac{1}{\sigma_{(M)}^{2}}\\
    \frac{\mu}{\sigma^{2}}&\triangleq\frac{\mu_{(0,1,...,M)}}{\sigma_{(0,1,...,M)}^{2}}=\frac{\mu_{(0)}}{\sigma_{(0)}^{2}}+\frac{\mu_{(1)}}{\sigma_{(1)}^{2}}+...+\frac{\mu_{(M)}}{\sigma_{(M)}^{2}}
\end{align*}
which is equivalent to \eqref{eq:varcom} and \eqref{eq:poscom}. Note that the scaling value
\begin{align}
    S =& \frac{(2\pi)^{-\frac{M}{2}} \sigma e^{({\mu^2}/{\sigma^{2}} - \sum_{i=0}^M  {\mu\left(Z\left(m,t\right)\right)^2}/{{\sigma\left(Z\left(m,t\right)\right)}^{2}} )/2}}{\prod_{m=0}^M\sigma\left(Z\left(m,t\right)\right)^2}
\end{align}
and it will not change the optimum. Thus, it is feasible to exclude this part from the calculation.

\end{appendices}

\bibliographystyle{IEEEtran}
\bibliography{reference/references.bib}

\end{document}